\documentclass{article}
\usepackage[utf8]{inputenc}
\usepackage{hw}
\usepackage{algorithm,bm}
\usepackage{lmodern}
\usepackage{amsmath}
\usepackage{comment}
\usepackage{natbib}
\usepackage{indentfirst,epstopdf}
\usepackage[margin=1in,top=1in,
            nohead,
            pdftex]{geometry}
\usepackage{url}
\usepackage{graphicx}
\usepackage{subcaption}
\usepackage[colorlinks=true,citecolor=blue,urlcolor=blue,linkcolor=blue]{hyperref}

\allowdisplaybreaks[4]

\def\det{{\rm det}}

\begin{document}
\title{Robust Sparse Precision Matrix Estimation and its Applications}
\author{Zhengke Lu \\
School of Mathematical Science, Nankai University\\
Long Feng \\
School of Statistics and Data Science, KLMDASR, LEBPS and LPMC, Nankai University}
\date{}
\maketitle

\begin{abstract}
We address the problem of robust sparse estimation of the precision matrix for heavy-tailed distributions in high-dimensional settings. In such high-dimensional contexts, we observe that the covariance matrix can be approximated by a spatial-sign covariance matrix, scaled by a constant. Based on this insight, we introduce two new procedures, the Spatial-Sign Constrained $l_1$ Inverse Matrix Estimation (SCLIME) and the Spatial-sign Graphic LASSO Estimation (SGLASSO), to estimate the precision matrix. Under mild regularity conditions, we establish that the consistency rate of these estimators matches that of existing estimators from the literature. To demonstrate its practical utility, we apply the proposed estimator to two classical problems: the elliptical graphical model and linear discriminant analysis. Through extensive simulation studies and real data applications, we show that our estimators outperforms existing methods, particularly in the presence of heavy-tailed distributions.

\end{abstract}

\section{Introduction}

The estimation of precision matrices, which represent the inverse of the covariance matrix, is a fundamental problem in high-dimensional statistics. In many applications, especially in fields like genomics, finance, and neuroimaging, the number of variables \( p \) often exceeds the number of observations \( n \), creating a high-dimensional setting where traditional methods fail due to singularity of sample covariance matrix. To address the challenges of high-dimensional precision matrix estimation, several methods have been proposed, each leveraging different assumptions and computational techniques. One such approach involves banding the Cholesky factor of the inverse covariance matrix, which assumes certain structural orderings in the precision matrix. These methods, often based on the assumption that the precision matrix exhibits some form of block structure or sparsity, have been explored by \cite{WuPourahmadi2003}, \cite{Huangetal2006}, and \cite{BickelLevina2008b}. These approaches focus on approximating the precision matrix by banding the Cholesky decomposition of its inverse, which reduces the dimensionality of the problem and allows for more efficient estimation. However, the success of these methods depends on the underlying sparsity structure being appropriately captured, and they are typically best suited for situations where the precision matrix has known block patterns or specific sparsity structures.

In contrast, penalized likelihood approaches have been introduced to estimate sparse precision matrices under more general conditions. These methods typically involve adding a penalty term to the likelihood function, which encourages sparsity in the precision matrix while maintaining a reasonable fit to the data. Among these, the \( l_1 \)-penalized normal likelihood estimator, commonly known as the graphical Lasso, has become one of the most widely studied and applied techniques. \cite{YuanLin2007} first proposed this approach, which applies an \( l_1 \)-norm penalty to the precision matrix, promoting sparsity by shrinking the off-diagonal entries towards zero. The graphical Lasso and its variants have been further studied and developed by \cite{dAspremontBanerjeeElGhaoui2008}, \cite{FriedmanHastieTibshirani2008}, and \cite{Rothmanetal2008}, among others.  \cite{cai2011constrainedl1minimizationapproach} introduced a novel method for inverse matrix estimation, termed constrained \( l_1 \)-minimization (CLIME), which is particularly well-suited for high-dimensional data. The method is computationally efficient, as it allows the estimation of the precision matrix one column at a time by solving a linear programming problem. The resulting estimator is constructed by combining the individual vector solutions, followed by a simple symmetrization procedure. Notably, CLIME does not require any outer iterations, making the algorithm highly scalable and suitable for large-scale applications.

In addition to the \( l_1 \)-penalized graphical Lasso, various extensions and alternative formulations of penalized likelihood methods have been explored. For instance, nonconvex penalties have been proposed to improve upon the potential bias of the \( l_1 \)-penalty, especially in cases where sparsity is expected but a more nuanced regularization is beneficial. \cite{LamFan2009} and \cite{FanFengWu2009} examined the use of the Smoothly Clipped Absolute Deviation (SCAD) penalty, a nonconvex penalty that reduces the bias introduced by the \( l_1 \)-penalty while still promoting sparsity. These methods are particularly useful in scenarios where the data are expected to follow a sparse model but may exhibit more complex structures or non-linear relationships.

Further developments in this area also include convergence rate analysis for these estimators. \cite{Yuan2009} derived the convergence rates of the \( l_1 \)-penalized estimator under the assumption that the data follow a sub-Gaussian distribution, while \cite{RavikumarWainwrightRaskuttiYu2008} analyzed the consistency of the \( l_1 \)-penalized precision matrix estimator under more stringent conditions, such as the mutual incoherence or irrepresentable conditions. Their work focused on deriving convergence rates in the elementwise \( l_\infty \)-norm and the spectral norm, providing rigorous guarantees for the performance of the estimator in high-dimensional settings. For a comprehensive overview of additional methods, please refer to the work by \cite{Cai2016}, \cite{fan2016overview}.

While the methods discussed above have provided robust and efficient tools for sparse precision matrix estimation in many settings, they are often built on assumptions that may not hold in real-world applications, such as the Gaussianity or sub-Gaussianity of the data. In practical scenarios, especially in fields like finance and genomics, data can exhibit heavy tails or other non-Gaussian characteristics, making it crucial to develop robust estimators. For robust estimation of precision matrices, a common approach involves replacing the sample covariance matrix in the objective function with robust pairwise covariance or correlation matrices. Notable works in this area include those by \cite{Ollerer2015}, \cite{Tarr2016}, \cite{Wegkamp2016}, \cite{Han2017}, \cite{Loh2018}. \cite{wang2023novel} propose a novel method that integrates the technique of modified Cholesky decomposition (MCD) with robust generalized M-estimators. The MCD reparameterizes the precision matrix and reformulates its estimation as a sequence of linear regression problems, allowing the straightforward incorporation of commonly used robust techniques. 

In this paper, we propose a novel robust estimator for the precision matrix based on the spatial-sign method, which has been widely applied in multivariate analysis for elliptical distributions. As demonstrated by \cite{oja2010multivariate}, the spatial-sign covariance matrix shares the same eigenvectors as the sample covariance matrix; however, the eigenvalues of the spatial-sign covariance matrix do not correspond directly to those of the covariance matrix. Nevertheless, under the assumption that the eigenvalues of the covariance matrix are bounded, we show that as the dimension tends to infinity, the spatial-sign covariance matrix approximates the covariance matrix multiplied by a constant. This result can be viewed as a "blessing of dimensionality."
Building on this, we estimate the precision matrix using the constrained \( l_1 \)-minimization (CLIME) method \citep{cai2011constrainedl1minimizationapproach} and graphic lasso \citep{YuanLin2007}. Under mild conditions, we establish the convergence rates for the resulting precision matrix estimator under various norms, which are consistent with the existing results in the literature.

Next, we consider two applications of sparse precision matrix estimation: the elliptical graphical model and linear discriminant analysis. The recovery of the support of the precision matrix is closely tied to the selection of graphical models. For Gaussian distributions, recovering the structure of the graph representing the conditional independence relationships between the components of the vector is equivalent to estimating the support of the precision matrix \citep{Lauritzen1996}. For elliptical distributions, all marginal and conditional distributions are also elliptical \citep{FangZhang1990}. As a result, partial uncorrelatedness implies conditional uncorrelatedness \citep{baba2004partial}. Therefore, estimating the support of the precision matrix for elliptical distributions is equivalent to recovering the structure of the graph representing conditional uncorrelation \citep{vogel2011elliptical,vogel2014robust}. For high-dimensional data, the estimation procedure for the precision matrix can be directly applied to the elliptical graphical model \citep{vogel2011elliptical}. In this paper, we introduce an additional thresholding step based on the precision matrix estimator and demonstrate its consistency.

Linear discriminant analysis (LDA) is a classical statistical method that aims to find a linear combination of features that best separates two or more classes in a dataset. In high-dimensional settings, where the number of features $p$ exceeds the number of observations $n$, traditional LDA is no longer directly applicable due to the singularity of the sample covariance matrix. This challenge has led to the development of various extensions of LDA that are specifically designed for high-dimensional data. One of the most well-known approaches is the regularized LDA, which introduces a penalty to shrink the covariance matrix, thus ensuring invertibility. For example, the method introduced by \cite{mclachlan2004discriminant} uses shrinkage techniques to stabilize the covariance estimation in high dimensions. In addition, \cite{chen2011high} proposed a high-dimensional version of LDA by employing a sparse precision matrix estimator, which encourages sparsity in the inverse covariance matrix, making the model more interpretable and computationally feasible. For a comprehensive overview of high-dimensional linear discriminant analysis, see \cite{zhang2020high}. In this paper, we also propose a robust linear discriminant analysis for elliptical distributions, built upon the sparse precision matrix estimator discussed earlier. The asymptotic properties of the proposed method are examined, and we provide results on its consistency and rate of convergence.

The paper is organized as follows. In Section 2, we present the estimation procedure for the precision matrix. Section 3 discusses two applications of the sparse precision matrix estimator. Simulation studies and real data applications are provided in Section 4. All proofs are presented in the Appendix.

{\it Notations}:  Throughout, for a vector ${\bm a} = (a_1, \ldots, a_p)^T \in \mathbb{R}^p$, we define 
$
\|{\bm a}\|_0 = \sum_{j=1}^p \1\{a_j\neq0\}, \|{\bm a}\|_1 = \sum_{j=1}^p |a_j|$ and $\|{\bm a}\|_2 = \sqrt{\sum_{j=1}^p a_j^2}.
$
For a matrix ${\bf A} = (a_{ij}) \in \mathbb{R}^{p \times q}$, we define the elementwise $\ell_\infty$ norm 
$
\|{\bf A}\|_\infty = \max_{1 \leq i \leq p, 1 \leq j \leq q} |a_{ij}|,
$
the operation norm 
$
\|{\bf A}\|_\op = \sup_{\|{x}\|_2 \leq 1} |{\bf A} {x}|_2,
$
the matrix $\ell_1$ norm 
$
\|{\bf A}\|_{L_1} = \max_{1 \leq j \leq q} \sum_{i=1}^p |a_{ij}|,
$
the matrix $\ell_\infty$ norm 
$
\|{\bf A}\|_{L_\infty} = \max_{1 \leq i \leq p} \sum_{j=1}^q |a_{ij}|,
$
the Frobenius norm 
$
\|{\bf A}\|_F = \sqrt{\sum_{i,j} a_{ij}^2},
$
and the elementwise $\ell_1$ norm 
$
\|{\bf A}\|_1 = \sum_{i=1}^p \sum_{j=1}^q |a_{ij}|.
$
$\id_p$ denotes a $p \times p$ identity matrix. For any two index sets $T$ and $T'$ and matrix ${\bf A}$, we use ${\bf A}_{T T'}$ to denote the $|T| \times |T'|$ matrix with rows and columns of ${\bf A}$ indexed by $T$ and $T'$, respectively. The notation ${\bf A} \succ 0$ indicates that ${\bf A}$ is positive definite. For any random variable $X \in \mathbb{R}$, we define the sub-Gaussian ($\|\cdot\|_{\psi_2}$) and sub-exponential norms ($\|\cdot\|_{\psi_1}$) of $X$ as follows: $
\|X\|_{\psi_2} := \sup_{k \geq 1} k^{-1/2} \left( \mathbb{E}|X|^k \right)^{1/k}
\quad \text{and} \quad
\|X\|_{\psi_1} := \sup_{k \geq 1} k^{-1} \left( \mathbb{E}|X|^k \right)^{1/k}.$ 
\section{Method}
\subsection{Setting}

Suppose $\bm X_1,\cdots,\bm X_n $ are i.i.d. copies of $\bm X\sim EC_p(\bm\mu,{\bf\Sigma}_0,r)$ i.e.
\begin{equation}\label{Setting}
    \bm X_i\deql \bm\mu+r_i{\bf\Sigma}_0^\frac{1}{2}\bm u_i
\end{equation}
where ${\bm u}_i\in \R^p$ is the corresponding i.i.d. copy of ${\bm u}\in\R^p$ , which is uniformly distributed on $\mathbb{S}^{p-1}$, and $r_i\in\R$ is the corresponding i.i.d. copy of $r\in\R$, which is a random variable with $\E[r_i^2]=p$. In addition, $r, u$ are independent. Here, we implicitly assume that $ {\bf\Sigma}_0$ is positive definite, and hence its square root exists, which is denoted by ${\bf\Sigma}_0^\frac{1}{2}$. To avoid confusion, we also denote ${\bf\Sigma}_0^\frac{1}{2}$ by ${\bf\Gamma}_0$. Here we define the {\it shape matrix} as follows:
\begin{equation}
    {\bf\Lambda}_0=\frac{p}{\tr({\bf\Sigma}_0)}{\bf\Sigma}_0
\end{equation}
Consequently, its inverse matrix can be defined as follows.
\begin{equation}\label{define V0}
    {\bf V}_0={\bf\Lambda}_0^{-1}=\frac{\tr({\bf\Sigma}_0)}{p}{\bf\Sigma}_0^{-1}
\end{equation}

In most of the literature, researchers are interested in estimating the precision matrix ${\bf\Omega}={\bf\Sigma}_0^{-1}$. In this paper, our focus is on estimating ${\bf V}_0$, which is proportional to ${\bf \Omega}$. This is because in numerous applications, an estimator of a constant multiple of the precision matrix suffices. For instance, in the elliptical graphical model, the support recovery of the precision matrix ${\bf\Omega}$ is equivalent to that of ${\bf V}_0$. Additionally, some literature directly assumes that $\text{tr}({\bf\Sigma}_0)=p$ for simplicity, as seen in \cite{han2016ecahighdimensionalelliptical}. In this scenario, ${\bf\Omega} = {\bf V}_0$. Hence, in the subsequent sections, we also refer to the matrix ${\bf V}_0$ as the precision matrix.

\subsection{Estimators}
If ${\bm X}\sim N_p({\bm\mu}, {\bf\Sigma})$, the famous Glasso method \citep{BanerjeeGhaouidAspremont2008,dAspremontBanerjeeElGhaoui2008,FriedmanHastieTibshirani2008}  for estimating the sparse precision matrix based on $l_1$ regulared log-determinant program solves the following optimization problem:
\begin{align}\label{glasso}
\hat{\bf \Omega}_{GLASSO}=\arg\min_{\Theta\succ 0} \tr({\bf\Theta \hat{\Sigma}})-\log \det({\bf \Theta})+\lambda_n \|{\bf\Theta}\|_1,
\end{align}
where $\bf{\hat{\Sigma}}$ is the sample covariance matrix and $\lambda_n$ is a tuning parameter. Here the constraint ${\bf\Theta} \succ 0$ ensures the matrix ${\bf\Theta}$ is positive definite. In a different approach, \cite{cai2011constrainedl1minimizationapproach} introduced the constrained
$l_1$-minimization for inverse matrix estimation (CLIME), which solves the following optimization problem:
\begin{align}\label{clime}
\min \|{\bf\Theta}\|_1, ~~~\text{subject to}~~\|\hat{\bf \Sigma}{\bf\Theta}-\id_p\|_\infty\le \lambda_n, {\bf\Theta}\in \mathbb{R}^{p\times p}.
\end{align}
Both of these methods are based on the sample covariance matrix, which may perform poorly under heavy-tailed distributions. In the realm of heavy-tailed distributions, the spatial-sign covariance matrix finds extensive application in robust statistical inference, as detailed in \cite{oja2010multivariate}. The spatial-sign covariance matrix \(S\) is defined as:
\[{\bf S} = \mathbb{E}\left(U({\bm X}_i-{\bm\mu})U({\bm X}_i-{\bm\mu})^{\top}\right),\quad\text{where}\quad U({\bm x})=\frac{\bm x}{\|{\bm x}\|}I({\bm x}\neq{\bm 0}).\]

Notably, while the eigenvectors of the spatial - sign covariance matrix coincide with those of the covariance matrix, their eigenvalues differ. According to Proposition 2.1 in \cite{han2016ecahighdimensionalelliptical}, if \(\text{rank}({\bf S}) = q\), the eigenvalues of \({\bf S}\) are given by:
\begin{align*}
\lambda_j(\mathbf{S})=E\left(\frac{\lambda_j(\boldsymbol{\Sigma}) Y_j^2}{\lambda_1(\boldsymbol{\Sigma}) Y_1^2+\cdots+\lambda_q(\boldsymbol{\Sigma}) Y_q^2}\right)
\end{align*}
where \(\boldsymbol{Y}:=(Y_1,\ldots,Y_q)^{\top}\sim N_q(\mathbf{0},\id_q)\) represents a standard \(q\)-dimensional multivariate Gaussian distribution. Here, \(\lambda_j({\bf A})\) denotes the \(j\)-th largest eigenvalue of the symmetric matrix \({\bf A}\).
In the context of fixed-dimensional data, the relationship between the spatial-sign covariance matrix and the covariance matrix is rather intricate. Recovering the covariance matrix from the spatial-sign covariance matrix is no straightforward task. However, for high-dimensional data, under the assumptions that the eigenvalues of the covariance matrix are uniformly bounded and the precision matrix exhibits sparsity, by Lemma \ref{error}, we can establish that \(\|p{\bf S}-{\bf\Lambda}_0\|_{\infty}=O(p^{-1/2})\).
Consequently, the precision matrix \({\bf V}_0\) satisfies the approximation \(p{\bf V}_0{\bf S}\approx\id_p\). This crucial result provides the impetus for leveraging the spatial-sign covariance matrix to estimate the precision matrix \({\bf V}_0\).

Following this insight, we propose a new estimator, the Spatial-Sign Constrained $l_1$ Inverse Matrix Estimator (SCLIME), motivated by the CLIME method (\ref{clime}). Let $\{\hat{\bf V}_1\}$ be the solution set of the following optimization problem:
\begin{equation}\label{obj}
    \min \|{\bf V}\|_1 \quad \text{subject to:} \quad \|p\hat {\bf S} \,{\bf V} - \id_p\|_\infty \leq \lambda_n, \quad {\bf V} \in \mathbb{R}^{p \times p}, 
\end{equation}
where $\lambda_n$ is a tuning parameter and $\hat {\bf S}$ is the  sample spatial-sign covariance matrix corresponding to ${\bm X}_1,\cdots,{\bm X}_n$. Specially, $\hat {\bf S}=\frac{1}{n}\sum_{i=1}^n U({\bm X}_i-\hat{\bm \mu})U({\bm X}_i-\hat{\bm \mu})^\sT$ and $\hat{\bm \mu}$ is the sample spatial-median, i.e.
\begin{align*}
\hat{\bm \mu}=\arg\min_{{\bm\mu} \in \mathbb{R}^p} \sum_{i=1}^n \|{\bm X}_i-{\bm\mu}\|_2.
\end{align*}
In \eqref{obj}, we do not impose the symmetry condition on $\bf V$, and as a result, the solution is not symmetric in general. The final SCLIME estimator of ${\bf V}_0$ is obtained by symmetrizing $\hat{\bf V}_1$ as follows. Write $\hat{\bf V}_1 = (\hat{v}_{ij}^1) = (\hat{v}_1^1, \ldots, \hat{v}_p^1)$. The SCLIME estimator $\hat{\bf V}_{SCLIME}$ of ${\bf V}_0$ is defined as
\begin{equation}\label{CSSIME}
    \hat{\bf V}_{SCLIME} = (\hat{v}_{ij}), \quad \text{where} \quad \hat{v}_{ij} = \hat{v}_{ji} = \hat{v}_{ij}^1 \1\{|\hat{v}_{ij}^1| \leq |\hat{v}_{ji}^1|\} + \hat{v}_{ji}^1 \1\{|\hat{v}_{ij}^1| > |\hat{v}_{ji}^1|\}. 
\end{equation}
Clearly, $\hat{\bf V}$ is now symmetric. We will show that with high probability, $\hat{\bf V}$ is also positive definite.

{
Similarly, we can also propose a spatial-sign Graphic Lasso estimator (SGLASSO) for precision matrix, 
\begin{align}\label{sglasso}
\hat{\bf V}_{SGLASSO}=\arg\min_{{\bf V} \succ 0} \tr(p\hat{\bf S}{\bf V} )-\log \det({\bf V})+\lambda_n \|{\bf V}\|_{1},
\end{align}
}

Thanks to the existing R packages, our proposed methods can be effortlessly implemented. First, the sample spatial-sign covariance matrix is estimated using the function \textit{SCov} in the R-package \texttt{SpatialNP}. Next, we utilize the function \textit{sugm} from the R-package \texttt{flare} to obtain the SCLIME estimator. Simultaneously, the function \textit{glassoFast} within the R-package \texttt{glassoFast} is employed to acquire the SGLASSO estimator. 

\subsection{Theoretical Results}
In this section, we investigate the theoretical properties of the
SCLIME and SGLASSO estimators and establish the rates of convergence under
different norms. We need the following assumptions:

\begin{assumption}{(Assumptions on $V_0$)}\label{V0}
\begin{subassumption}{\it (SCLIME)}\label{V0a}
    $\exists T>0, \,\,0\leq q<1,\,\,s_0(p)>0, s.t.$
(1) $\|{\bf V}_0\|_{L_1}\leq T$,
(2) $\max_{1\leq i\leq p}\sum_{j=1}^p|{ v}_{ij}|^q\leq s_0(p)$.
\end{subassumption}
\begin{subassumption}{\it (SGLASSO)}\label{V0b}
Let $d\triangleq \max_{i}\left|\{j: {\bf V}_{ij}\neq 0\}\right| $, $ E \triangleq \{(i,j)\in[p]\times[p] \quad|\quad i\neq j, {\bf V}_{ij}\neq 0\} $, $S\triangleq E \cup\{(1,1),\cdots,(p,p)\}$
and ${\bf\Gamma}^*\triangleq{\bf\Lambda}_0\otimes {\bf\Lambda}_0$. Denote $\|({\bf\Gamma}^*_{SS})^{-1}\|_{L_\infty}$ by $K_{\Gamma^*}$ and $\|
{\bf\Lambda}_0\|_{L_\infty}$ by $K_{\Lambda_0}$. We assume (1) $|E|\leq s=O(p)$, $d=O(1)$. (2) $\exists \alpha\in (0,1), s.t.\|{\bf\Gamma}^*_{S^cS}{\bf\Gamma}^*_{SS}\|_{L_{\infty}}\leq 1-\alpha$. (3) $\exists K^*=O(1)>0, s.t. \max\{K_{\Gamma^*},K_{\Lambda_0}\}\leq K^*$.
\end{subassumption}
\end{assumption}
\begin{assumption}{(Bound of eigenvalues of covariance matrix) }\label{lambda}
$\exists \eta>0, s.t.$
$
   {\eta}<\lambda_p({\bf\Sigma}_0)\leq\lambda_1({\bf \Sigma}_0)<\frac{1}{\eta}.
$

\end{assumption}
As a result, we can also derive the following desirable property of $\Sigma_0$, we will use it directly as a part of Assumption \ref{lambda}.
\begin{itemize}
    \item (Bound of covariance matrix) $\exists M>0, s.t.$
$
 \|{\bf\Sigma}_0\|_\infty<M.
$
\item(Order of trace) $
\tr({\bf\Sigma}_0)\asymp p.
$
Particularly, we assume that $\exists c_0>0, \,\forall p>0$, 
$
c_0p\leq\tr({\bf\Sigma}_0)\leq \frac{1}{c_0}p.
$
\end{itemize}

\begin{assumption}\label{A1A2}
    Define 
$
\zeta_k = \mathbb{E}(\xi_i^{-k}), \quad \xi_i = \|{\bm X}_i - {\bm\mu}\|_2, \quad \nu_i = \zeta_1^{-1} \xi_i^{-1}.
$
\begin{enumerate}
    \item[(1)] $\zeta_k \zeta_1^{-k} < \zeta \in (0, \infty)$ for $k = 1, 2, 3, 4$ and all $d$.
    \item[(2)] $\limsup_p \|{\bf S}\|_{\op} < 1 - \psi < 1$ for some positive constant $\psi$.
    \item[(3)] $\nu_i$ is sub-gaussian distributed, i.e. $\|\nu_i\|_{\psi_2}\leq K_{\nu}<\infty$.
\end{enumerate}
\end{assumption}
Assumption \ref{V0a} aligns with the uniformity class of matrices \(\mathcal{U}(q,s_0(p))\) introduced in \cite{cai2011constrainedl1minimizationapproach}, and is similarly considered in \cite{BickelLevina2008b} for covariance matrix estimation. Assumption \ref{V0b} aligns with the  conditions assumed in \cite{RavikumarWainwrightRaskuttiYu2008}. Assumption \ref{lambda} is a standard condition in high-dimensional data analysis, appearing in works such as \cite{avella2018robust} and \cite{BickelLevina2008a}. Finally, Assumptions \ref{A1A2} are consistent with conditions (A1-A2) in \cite{feng2024spatialsignbasedprincipal}, which ensure the consistency of the spatial median estimator.

\begin{theorem}\label{main}
Under Assumption \ref{V0a}, \ref{lambda},  and \ref{A1A2}, $\hat {\bf V}_{SCLIME}$ defined in \eqref{CSSIME} satisfies the following property. There exist constants $C_{c_0,\eta,T,M}$ and $C$, such that when $n, p$ is sufficiently large that satisfy
$
    n>\left(3\log p\right)^{\frac{1}{3}}
$, if we pick $\lambda_n = T\left(\frac{\sqrt{2}C\left(8 +\eta^2C_{c_0,\eta,T,M }\right) }{\eta^2}
\sqrt{\frac{\log p }{n}}+\frac{C_{c_0,\eta,T,M }}{\sqrt{p}}\right)$,  with probability larger than $1- 2p^{-2}$, we have
\begin{align}
    & \|\hat {\bf V}_{SCLIME}-{\bf V}_0\|_\infty\leq 4T^2\left(\frac{\sqrt{2}C(8 +\eta^2C_{c_0,\eta,T,M }) }{\eta^2}
\sqrt{\frac{\log p }{n}}+\frac{C_{c_0,\eta,T,M }}{\sqrt{p}}\right)\label{inftyerror}\\
    & \|\hat {\bf V}_{SCLIME}-{\bf V}_0\|_\op\leq \|\hat {\bf V}_{SCLIME}-{\bf V}_0\|_{L_1}\leq C_4\left(\frac{\sqrt{2}C\left(8 +\eta^2C_{c_0,\eta,T,M }\right) }{\eta^2}
    \sqrt{\frac{\log p }{n}}+\frac{C_{c_0,\eta,T,M }}{\sqrt{p}}\right)^{1-q} s_0(p) \label{L1error}\\
    & \frac{1}{p}\|\hat {\bf V}_{SCLIME}-{\bf V}_0\|_{F}^2\leq C_5 \left(\frac{\sqrt{2}C \left(8 +\eta^2C_{c_0,\eta,T,M }\right)}{\eta^2}
\sqrt{\frac{\log p }{n}}+\frac{C_{c_0,\eta,T,M }}{\sqrt{p}}\right)^{2-q}s_0(p) \label{Ferror}
\end{align}
    where $C_4\leq\left(1 + 2^{1-q} + 3^{1-q}\right)\left(4T^2\right)^{1-q}$ and $C_5\leq 4T^2C_4 $.
\end{theorem}
We prove Theorem \ref{main} in Appendix \ref{appendix: proof main}. The key idea of the proof is to bound the error $\|p\hat {\bf S}-{\bf \Lambda}_0\|_\infty$, which is equivalent to bound two separate parts, approximation error $\|p{\bf S}-{\bf \Lambda}_0\|_\infty$ and estimation error $\|\hat {\bf S}-{\bf S}\|_\infty$. By Lemma \ref{error}, we prove that for elliptical family, the approximation error is of the order of $O\left(p^{-\frac{1}{2}}\right)$ under mild conditions. By Lemma \ref{SSestimation}, we prove that for elliptical family, the estimation error is of the order of $O\left(\sqrt{\frac{\log p}{n}}\right)$ under mild conditions. Combining them together, we can show that $ \|\hat {\bf V}_{SCLIME}-{\bf V}_0\|_\infty$ achieves the rate of $O\left(\sqrt{\frac{\log p}{n}}+\frac{1}{\sqrt{p}}\right)$, which matches with the existing estimators from the literature, see \cite{cai2011constrainedl1minimizationapproach}, if $p\log p/n\to \infty$. These results imply that when the dimension of the data is either significantly larger than or comparable to the sample size, the approximation error between the covariance matrix and the product of the spatial-sign covariance matrix and $p$ will not impact the convergence rate of the precision matrix estimation. Moreover, as the dimension increases, the influence of this approximation error becomes even more negligible.  

\begin{theorem}\label{maing}
Under Assumption \ref{V0b}, \ref{lambda} and \ref{A1A2}, $\hat {\bf V}_{SGLASSO}$ defined in \eqref{sglasso} satisfies the following property. There exist constants $C_{c_0,\eta,s,M}$ and $C$ , such that when $n, p$ is sufficiently large that satisfy \[\left(\frac{\sqrt{2}C\left(8 +\eta^2C_{c_0,\eta,s,M }\right) }{\eta^2}
\sqrt{\frac{\log p }{n}}+\frac{C_{c_0,\eta,s,M }}{\sqrt{p}}\right)\leq \min\left\{\frac{1}{6(K^*)^5d\left(1+\frac{8}{\alpha}\right)^2},\frac{1}{6(K^*)^2d\left(1+\frac{8}{\alpha}\right)}\right\}\] and $n>\left(3\log p\right)^\frac{1}{3}$,   if $\lambda_n = \frac{8}{\alpha}\left(\frac{\sqrt{2}C \left(8 +\eta^2C_{c_0,\eta,s,M }\right)}{\eta^2}
\sqrt{\frac{\log p }{n}}+\frac{C_{c_0,\eta,s,M }}{\sqrt{p}}\right)$ , with probability larger than $1-2p^{-2}$, we have
\begin{align}
   & \|\hat {\bf V}_{SGLASSO}-{\bf V}_0\|_\infty\leq 2K^*\left(\frac{8}{\alpha}+1\right)\left(\frac{\sqrt{2}C( 8 +\eta^2C_{c_0,\eta,s,M })}{\eta^2}
\sqrt{\frac{\log p }{n}}+\frac{C_{c_0,\eta,s,M }}{\sqrt{p}}\right)\label{inftyerrorglasso}\\
& \|\hat {\bf V}_{SGLASSO}-{\bf V}_0\|_\op\leq \|\hat {\bf V}_{SGLASSO}-{\bf V}_0\|_{L_1}\leq 2K^*d\left(\frac{8}{\alpha}+1\right)\left(\frac{\sqrt{2}C(8 +\eta^2C_{c_0,\eta,s,M }) }{\eta^2}
\sqrt{\frac{\log p }{n}}+\frac{C_{c_0,\eta,s,M }}{\sqrt{p}}\right)\label{L1errorglasso}\\
& \frac{1}{p}\|\hat {\bf V}_{SGLASSO}-{\bf V}_0\|_{F}^2\leq 4{K^*}^2\left(\frac{s}{p}+1 \right)\left(\frac{8}{\alpha}+1\right)^2 \left(\frac{\sqrt{2}C(8 +\eta^2C_{c_0,\eta,s,M }) }{\eta^2}
\sqrt{\frac{\log p }{n}}+\frac{C_{c_0,\eta,s,M }}{\sqrt{p}}\right)^2\label{Ferrorglasso}
\end{align}

\end{theorem}
We prove Theorem \ref{maing} in Appendix \ref{appendix: proof maing}. The proof is based on \cite{RavikumarWainwrightRaskuttiYu2008}, which mainly relies on (3) in Assumption \ref{V0b}, i.e. the assumption that two key quantities $\|{\bf\Lambda}_0\|_{L_\infty}$ and $\|({\bf\Gamma}_{SS}^*)^{-1}\|_{L_\infty}$ can be bounded by some absolute constants. In fact, this assumption can be discarded for Frobenius norm, referring to the proof in \cite{Rothmanetal2008} with slight modification. Overall, The rate achieved by SGLASSO is of the same order as the rate achieved by SCLIME, but with stronger assumptions.

\section{Application}
\subsection{Elliptical Graphical Model}  

Graphical models serve as powerful tools for deciphering the dependence structure inherent in data. In particular, the support of the conditional correlation matrix provides a means to characterize the linear dependence structure within a dataset. As per the findings in \cite{baba2004partial}, for data that follows an elliptical distribution, the property of partial uncorrelatedness directly implies conditional uncorrelatedness. This key result allows us to focus solely on the support of the partial correlation matrix when dealing with elliptically distributed data.

To elaborate further, consider a random vector ${\bm X}=(x_1,\cdots,x_p)^\top \sim EC_p({\bm\mu},{\bf\Sigma}_0,r)$. The partial correlation between variables $x_i$ and $x_j$, conditional on all other variables (i.e., given ${\bm X}_{\backslash{(i,j)}}$), is precisely equal to the $(i,j)$-th entry of the matrix $-\diag\{{\bf\Sigma}_0^{-1}\}^{-\frac{1}{2}}{\bf\Sigma}_0^{-1} \,\diag\{{\bf\Sigma}_0^{-1}\}^{-\frac{1}{2}}$ \citep{whittaker2009graphical}. When examining the support, or delving even deeper into the signs of the elements within the matrix, the precision matrix, up to a scaling constant, emerges as the matrix of primary interest. This highlights that $V_0$ is indeed one of the matrices central to our analysis. Consequently, the proposed estimator can be effectively employed to reconstruct the elliptical graphical model, offering valuable insights into the intricate relationships among the variables in the dataset.

First, for SCLIME, We introduce an additional thresholding step based on $\hat{\bf V}_{SCLIME}$. More specifically, define a threshold estimator $\tilde{\bf V}_{SCLIME} = (\tilde{v}_{ij})$ with
\[
\tilde{v}_{ij} = \hat{v}_{ij}\1\{|\hat{v}_{ij}| \geq \tau_n\},
\]
where $\tau_n > 4T\lambda_n$ is a tuning parameter and $\lambda_n= T\left(\frac{\sqrt{2}\left(8 +\eta^2C_{c_0,\eta,T,M }\right)C }{\eta^2}
\sqrt{\frac{\log p }{n}}+\frac{C_{c_0,\eta,T,M }}{\sqrt{p}}\right)$.

Define
\begin{align*}
\mathcal{M}(\tilde{\bf V}_{SCLIME}) &= \{\text{sgn}(\tilde{v}_{ij}), \, 1 \leq i, j \leq p\},\\
\mathcal{M}({\bf V}_0) &= \{\text{sgn}(v^0_{ij}), \, 1 \leq i, j \leq p\},\\
S({\bf V}_0)& = \{(i, j) : v^0_{ij} \neq 0\},\\
\theta_{\min} &= \min_{(i, j) \in S({\bf V}_0)} |v^0_{ij}|.
\end{align*}
\begin{theorem}\label{Graphthm}
Under Assumption \ref{V0a}, \ref{lambda},  and \ref{A1A2}, if $\lambda_n= T\left(\frac{\sqrt{2}C\left(8 +\eta^2C_{c_0,\eta,T,M }\right) }{\eta^2}
\sqrt{\frac{\log p }{n}}+\frac{C_{c_0,\eta,T,M }}{\sqrt{p}}\right)$ and  $\theta_{\min}>2\tau_n $, then when $n, p $ is sufficiently large that satisfy $n>\left(3\log p\right)^{\frac{1}{3}}$, with probability larger than $1-2p^{-2}$,
\[
\mathcal{M}(\tilde{\bf V}_{SCLIME})=\mathcal{M}({\bf V}_0).
\]
\end{theorem}
We present the proof of Theorem \ref{Graphthm} in Appendix \ref{appendix: proof graph}. The threshold estimator $\tilde{V}_{SCLIME}$ not only successfully recovers the sparsity pattern of the precision matrix $V_0$, but also accurately determines the signs of its nonzero elements. This characteristic is often referred to as sign consistency in certain statistical literature. Analogous to the approach in \cite{cai2011constrainedl1minimizationapproach}, the condition $\theta_{\min}>2\tau_n$ is essential for ensuring that the nonzero elements of the precision matrix are correctly identified and retained. In other words, this condition plays a pivotal role in safeguarding the integrity of the estimated precision matrix structure. Furthermore, under the assumption that $p\log p/n\to \infty$ and the parameters $T$ and $\eta$ are independent of the sample size $n$ and the dimension $p$, the quantity $\tau_n$ is of the order $\sqrt{\log p / n}$. This order is comparable to the one assumed in \cite{RavikumarWainwrightRaskuttiYu2008} for distributions with exponential-type tails. However, a key distinction lies in our assumption that the underlying model follows an elliptical distribution. This assumption is notably more general and less restrictive than the ones made in the aforementioned reference, thereby enhancing the applicability and robustness of our theoretical framework. 

Similarly, for SGLASSO, we redefine $\tau_n> 2K^*\left(1+\frac{\alpha}{8}\right)\lambda_n$, where $\lambda_n=\frac{8}{\alpha}\left(\frac{\sqrt{2}C\left(8 +\eta^2C_{c_0,\eta,s,M }\right) }{\eta^2}
\sqrt{\frac{\log p }{n}}+\frac{C_{c_0,\eta,s,M }}{\sqrt{p}}\right)$. Let $\mathcal{M}(\tilde{\bf V}_{SGLASSO})$ be the analog of $\mathcal{M}(\tilde{\bf V}_{CLIME})$, we can derive a similar result.
\begin{theorem}\label{Graphthmg}
    Under Assumption \ref{V0b}, \ref{lambda},  and \ref{A1A2}, if $\lambda_n= \frac{8}{\alpha}\left(\frac{\sqrt{2}C\left(8 +\eta^2C_{c_0,\eta,s,M }\right) }{\eta^2}
\sqrt{\frac{\log p }{n}}+\frac{C_{c_0,\eta,s,M }}{\sqrt{p}}\right)$ and  $\theta_{\min}>2\tau_n $, then when $n, p $ is sufficiently large that satisfy 
\[
    \left(\frac{\sqrt{2}C\left(8 +\eta^2C_{c_0,\eta,s,M }\right) }{\eta^2}
\sqrt{\frac{\log p }{n}}+\frac{C_{c_0,\eta,s,M }}{\sqrt{p}}\right)\leq \min\left\{\frac{1}{6(K^*)^5d\left(1+\frac{8}{\alpha}\right)^2},\frac{1}{6(K^*)^2d\left(1+\frac{8}{\alpha}\right)}\right\}
\] and $n>\left(3\log p\right)^{\frac{1}{3}}$, with probability larger than $1-2p^{-2}$,
\[
\mathcal{M}(\tilde{\bf V}_{SGLASSO})=\mathcal{M}({\bf V}_0).
\]
\end{theorem}
The proof of Theorem \ref{Graphthmg} closely mirrors that of Theorem \ref{Graphthm}. Due to this similarity, in order to avoid redundancy and keep the appendix concise, we have chosen to omit the proof of Theorem \ref{Graphthmg} from the appendix.  

\subsection{Linear Discriminant Analysis}\label{LDA}
Linear Discriminant Analysis (LDA) holds a pivotal position in a diverse range of real-world applications, spanning fields such as pattern recognition, machine learning, and data classification. Its effectiveness lies in its ability to find the optimal linear combination of features that maximally separates different classes, thereby facilitating accurate classification and prediction. The classic Fisher discriminant rule, a fundamental component of LDA, has traditionally been applied within the Gaussian setting. Under the assumption of Gaussian distributions for the data belonging to different classes, the Fisher discriminant rule provides a powerful and well-understood framework for classification.  However, traditional linear discriminant analysis (LDA) cannot be directly applied to high-dimensional data due to the non-invertibility of the sample covariance matrix. A common approach is to replace the sample covariance matrix in LDA with a sparse covariance or precision matrix estimator, as proposed in \cite{Rothmanetal2008}, \cite{whittaker2009graphical}, and \cite{Shao2011}. Nevertheless, these methods rely on Gaussian or sub-Gaussian assumptions, which may not be robust when dealing with heavy-tailed distributions.

Notably, \cite{Fang1990} extended the applicability of the Fisher discriminant rule beyond the Gaussian case. Their research demonstrated that Fisher's rule remains optimal even when the random vectors are drawn from elliptical distributions. Elliptical distributions are a broader class of distributions that encompass the Gaussian distribution as a special case, and they exhibit a high degree of flexibility in modeling real-world data. This finding significantly expands the scope of LDA, enabling its application to a wider variety of datasets that may not conform to the strict Gaussian assumptions. Next, we focus on high-dimensional discriminant analysis under the elliptical distribution.

Suppose we have samples from $({\bm X},Y)$, where ${\bm X}\in \R^p, Y\in \{0,1\}$. Suppose $\P(Y=1)=1-\P(Y=0)=p_1$, and 
\[
    {\bm X}|Y=0 \sim EC_p({\bm\mu}_0,{\bf \Sigma}_0,r)
\]
\[
    {\bm X}|Y=1 \sim EC_p({\bm \mu}_1,{\bf \Sigma}_0,r)
\]
Any linear discriminant rule has the form,
\[
    \delta_{\bm w}({\bm X})=\1\{{\bm w}^\sT ({\bm X}-{\bm \mu}_a)>0\}
\]
where ${\bm \mu}_a = \frac{{\bm \mu}_1+{\bm\mu}_0}{2}$ and when $\delta_{\bm w}({\bm X})=0$ or  $1$, it means the rule chooses class "$0$" for ${\bm X}$ or "$1$", respectively. We also set ${\bm \mu}_d = \frac{{\bm \mu}_1-{\bm \mu}_0}{2}$.

According to \cite{https://doi.org/10.1111/j.1540-6261.1983.tb02499.x}, we know the linear transformation of an elliptical distribution is still an elliptical distribution. To be more specific, suppose ${\bm X}\sim EC_p({\bm \mu},{\bf \Sigma}_0,r)$ and ${\bf D}\in \R^{q\times p}$ is a matrix with full row rank, then ${\bf D}{\bm X}\sim EC_q({\bf D}{\bm \mu},{\bf D\Sigma_0 D}^\sT,r^\prime)$. Therefore, we can calculate the misclassification rate of a linear discriminant rule as a function of the projection vector,
\begin{align*}
    R(\delta_{\bm w})=&\P(X=0|Y=1)\P(Y=1)+\P(X=1|Y=0)\P(Y=0)\\
    =&p_1\P_1({\bm w}^\sT ({\bm X}-{\bm \mu}_1)\leq -{\bm w}^\sT {\bm \mu}_d)+(1-p_1)\P_0({\bm w}^\sT({\bm X}-{\bm \mu}_0)\geq {\bm w}^\sT{\bm \mu}_d)\\
    =& 1-\Psi\left(\frac{{\bm w}^\sT {\bm\mu}_d}{\sqrt{{\bm w}^\sT{\bf \Sigma}_0 {\bm w}}}\right)
\end{align*}
where $\Psi$ is the c.d.f. of $EC_1(0,1,r^*)$. 

In general, there is no closed-form for $\Psi$ as a function of the distribution of ${\bm X}$. But in some classical cases, such as multivariate Gaussian and multivariate student t distribution, $\Psi$ has a clear form. Specifically, when ${\bm X}\sim N({\bm \mu},{\bf \Sigma}_0)$, $\Psi$ is the c.d.f. of the standard Gaussian distribution and when ${\bm X}\sim t_\nu({\bm \mu},{\bf \Sigma}_0)$, $\Psi$ is the c.d.f. of the $t_\nu(0,1)$.

Now, let us take a closer look at $R(\delta_{\bm w})$. Since $\Psi$ is monotonically increasing, just as the same argument for Fisher discriminant rule in Gaussian setting, we can know the best classifier is ${\bm w}_\infty={\bf \Sigma}_0^{-1}{\bm \mu}_d$. An important observation is that the rule is the same up to a scaling constant, i.e. $c{\bm w}$ and ${\bm w}$ yields the same decision boundary. Therefore, the better we can estimate ${\bf \Sigma}_0^{-1}$ and ${\bm \mu}_d$ up to a scaling constant, the more closely we can approach the best classifier. Hence, SCLIME and SGLASSO estimators can be applied here as consistent estimators of the best projection vector. 

To be more specific, let us first assume we have $n_0$ samples from ${\bm X}|Y=0$ and $n_1$ samples from ${\bm X}|Y=1$. Define $\hat {\bf S}({\bm X}|Y=k), k=0,1$ be the sample spatial-sign covariance matrix of the samples from ${\bm X}|Y=k$, respectively. Let
\begin{align*}
\hat{\bf S}_{pool}=\frac{n_0}{n_0+n_1}\hat {\bf S}(X|Y=0)+\frac{n_1}{n_0+n_1}\hat {\bf S}(X|Y=1)
\end{align*}
be the pool sample spatial-sign covariance matrix of these two samples. And we define the corresponding SCLIME and SGLASSO estimator, $\hat {\bf V}_{SCLIME}$ and $\hat {\bf V}_{SGLASSO}$
be the estimators of ${\bf V}_0$ defined in \eqref{define V0} based on the pool sample spatial-sign covariance matrix $\hat{\bf S}_{pool}$.  And let  $\hat {\bm \mu}_d=(\hat {\bm \mu}(X|Y=1)-\hat{\bm \mu}({\bm X}|Y=0))/2$ be the estimator of ${\bm \mu}_d$ where $\hat{\bm \mu}({\bm X}|Y=k), k=0,1$ is the sample spatial-median of the samples from ${\bm X}|Y=k$, respectively. 

For SCLIME, similar as the major argument of Theorem \ref{main}, we can derive a similar bound on $\|\hat {\bf V}_{SCLIME}-{\bf V}_0\|_\infty$ only with $n$ substituted by $n_0$ and $n_1$. Since they are of the same order, to simplify the description, let us consider the case when $p_1=0.5$ and assume that $n_0\approx n/2, n_1\approx n/2$ for the calculation. 
\begin{theorem}\label{LDATheorem}
    Suppose $\Sigma_0$ meets Assumption \ref{V0a}, \ref{lambda}, and Assumption \ref{A1A2} holds, if $\|{\bm \mu}_d\|_1$ is bounded by an absolute constant $M^*$, and if $\lambda_n=T\left(\frac{2C\left(8 +\eta^2C_{c_0,\eta,T,M }\right) }{\eta^2}
\sqrt{\frac{\log p }{n}}+\frac{C_{c_0,\eta,T,M }}{\sqrt{p}}\right)$, then for sufficiently large $n, p$ that satisfy $n>(3\log p)^\frac{1}{3}$, with probability greater than $1-4p^{-2}$, 
    \[
            \|\hat {\bf V}_{SCLIME}\hat {\bm \mu}_d-{\bf V}_0{\bm \mu}_d\|_\infty
     \leq \left(\frac{2\sqrt{6} TC_*}{\zeta_1 \eta^2} + \frac{8M^*T^2C_*C(8 +\eta^2C_{c_0,\eta,T,M }) }{\eta^2}\right)
\sqrt{\frac{\log p }{n}}+\frac{4M^*T^2C_*C_{c_0,\eta,T,M }}{\sqrt{p}}
    \]
    where $\zeta_1 $ is defined in Assumption $\ref{A1A2}$, $C_*$ is an absolute constant.
\end{theorem}
We prove Theorem \ref{LDATheorem} in Appendix \ref{appendix: proof LDA}. Similarly, for SGLASSO, we also  assume that $n_0\approx n/2, n_1\approx n/2$ for the calculation. 
\begin{theorem}\label{LDATheoremg}
       Suppose $\Sigma_0$ meets Assumption \ref{V0b}, \ref{lambda}, and Assumption \ref{A1A2} holds, if $\|{\bm \mu}_d\|_1$ is bounded by an absolute constant $M^*$, and if $\lambda_n=\frac{8}{\alpha}\left(\frac{2C\left(8 +\eta^2C_{c_0,\eta,s,M }\right) }{\eta^2}
\sqrt{\frac{\log p }{n}}+\frac{C_{c_0,\eta,s,M }}{\sqrt{p}}\right)$, then for sufficiently large $n, p$ that satisfy\[
    \left(\frac{2C\left(8 +\eta^2C_{c_0,\eta,s,M }\right) }{\eta^2}
\sqrt{\frac{\log p }{n}}+\frac{C_{c_0,\eta,s,M }}{\sqrt{p}}\right)\leq \min\left\{\frac{1}{6(K^*)^5d\left(1+\frac{8}{\alpha}\right)^2},\frac{1}{6(K^*)^2d\left(1+\frac{8}{\alpha}\right)}\right\}
\] and  $n>(3\log p)^\frac{1}{3}$, with probability greater than $1-4p^{-2}$ 
    \[
            \|\hat {\bf V}_{SGLASSO}\hat {\bm\mu}_d-{\bf V}_0{\bm \mu}_d\|_\infty \le
     \left(\frac{2\sqrt{6} C_{**}}{c_0\zeta_1 \eta^3} + \frac{4M^*K^*2C_{**}C\left(\frac{8}{\alpha}+1\right)\left(8 +\eta^2C_{c_0,\eta,s,M }\right) }{\eta^2}\right)
\sqrt{\frac{\log p }{n}}+\frac{2M^*K^*C_{**}C_{c_0,\eta,s,M }\left(1+\frac{8}{\alpha}\right)}{\sqrt{p}}
    \]
    where $\zeta_1 $ is defined in Assumption $\ref{A1A2}$, $C_{**}$ is an absolute constant.
\end{theorem}
The proof of Theorem \ref{LDATheoremg} is also an analog of the proof of Theorem \ref{LDATheorem}, so we omit it in the appendix, too.

From the above theorems, we can see that  both SCLIME and SGLASSO can be utilized to generate consistent estimators for the Fisher discriminant rule with rate $O\left(\sqrt{\frac{\log p }{n}}+\frac{1}{\sqrt{p}}\right)$, which can be viewed as a successful approach to handle the challenge of estimating inverse of covariance matrix when $p>n$, according to \cite{190c5a0d-ed5d-320b-82c0-14df50eada30}. 
\section{Numerical Experiment}
\subsection{Simulation}
\subsubsection{Precision Matrix Estimation}\label{simu: pme}
We consider three models as follows:
\begin{itemize}
    \item[(I)] ${\bf \Omega}=(0.6^{|i-j|})_{1\le i,j\le p}$;
    \item[(II)] The second model comes from \cite{Rothmanetal2008}. We let $\boldsymbol{\Omega}=\mathbf{B}+\delta \mathbf{I}$, where each off-diagonal entry in $\mathbf{B}$ is generated independently and equals 0.5 with probability 0.1 or 0 with probability $0.9 . \delta$ is chosen such that the conditional number (the ratio of maximal and minimal singular values of a matrix) is equal to $p$. Finally, the matrix is standardized to have unit diagonals.
    \item[(III)] ${\bf \Sigma}=(0.6^{|i-j|})_{1\le i,j\le p}$ and ${\bf \Omega}={\bf \Sigma}^{-1}$.
\end{itemize}
Model I has a banded structure, and the values of the entries
decay as they move away from the diagonal. Model II is an example of a sparse matrix without any special sparsity patterns. In Model III, ${\bf \Sigma}$ can be well approximated by
a sparse matrix and the inverse ${\bf \Omega}$ is a 3-sparse matrix.

Three elliptical distributions are considered:
\begin{itemize}
    \item Multivariate normal distribution: $N(0,{\bf\Sigma})$, 
    \item Multivariate t-distribution: $t_3(0,{\bf\Sigma})/\sqrt{3}$, 
    \item Mixed multivariate normal distribution: $MN(0.8,3,{\bf\Sigma})/\sqrt{2.6}$.
\end{itemize}
Here $t_p(0, {\bf \Lambda}, v)$ denotes a $p$-dimensional t-distribution with degrees of freedom $v$ and scatter matrix $\Lambda$. ${MN}(\gamma, \sigma, {\bf \Lambda})$ refers to a mixture multivariate normal distribution with density function $(1-\gamma) f_p(0, {\bf \Lambda}) + \gamma f_p(0, \sigma^2 {\bf \Lambda})$, where $f_p(a, b)$ is the density function of the $p$-dimensional normal distribution with mean $a$ and covariance matrix $b$.
We consider sample size $n=100$ and four different dimensions $p=30,60,90,120$. All the results are based on 100 replications.

The tuning parameter $\lambda_n$ is selected the same as \cite{cai2011constrainedl1minimizationapproach}. For the method CLIME and GLASSO, we generate a training sample of size $n=100$ from the same elliptical distribution with mean $0$ and covariance matrix ${\bf \Sigma}_0$, and an independent sample of size 100 from the same distribution for validating the tuning parameter $\lambda$. Using the training data, we compute a series of estimators with 50 different values of $\lambda$ and use the one with the smallest likelihood loss on the validation sample, where the likelihood loss is defined by
$$
L(\boldsymbol{\Sigma}, \boldsymbol{\Omega})=\langle\boldsymbol{\Omega}, \boldsymbol{\Sigma}_n\rangle-\log \operatorname{det}(\boldsymbol{\Omega})
$$
where ${\bf \Sigma}_n$ is the sample covariance matrix of the validation sample. While for the SCLIME and SGLASSO, we use a similar likelihood loss:
$$
L(\boldsymbol{S}, \boldsymbol{\Omega})=\langle\boldsymbol{\Omega}, p\boldsymbol{S}_n\rangle-\log \operatorname{det}(\boldsymbol{\Omega})
$$
where ${\bf S}_n$ is the sample spatial-sign covariance matrix of the validation sample. 

We first measure the estimation quality by the following matrix norms: the operator norm, the matrix $l_1$ norm, and the Frobenius norm. Table \ref{tab1}-\ref{tab3} report the averages and standard errors of these losses of Model I-III, respectively.
From the three tables, we can observe that the spatial-sign-based estimators, namely SCLIME and SGLASSO, exhibit similar performance to the traditional CLIME and GLASSO methods when the data follows a normal distribution. However, in the case of heavy-tailed distributions, SCLIME and SGLASSO significantly outperform CLIME and GLASSO, demonstrating their superior robustness and ability to handle non-normal data effectively. Additionally, the standard deviation of SCLIME and SGLASSO is smaller than that of CLIME and GLASSO in most cases, indicating that our proposed methods exhibit greater stability.  This suggests that the spatial-sign-based methods provide a more reliable estimation approach in situations where the underlying data deviates from normality, especially in the presence of  heavy-tailed behavior.

\begin{table}[htbp]
\caption{Comparisons of average (standard deviation) matrix losses for three distributions under Model I over 100 replications}
\centering
{\tiny
	
 \setlength{\tabcolsep}{3pt}
 \renewcommand{\arraystretch}{2}
	\begin{tabular}{l|cccc|cccc|cccc} \hline
 &\multicolumn{4}{c}{Normal Distribution} &\multicolumn{4}{c}{$t_3$ Distribution} &\multicolumn{4}{c}{Mixture Normal Distribution}\\ \hline
$p$&CLIME&SCLIME&GLASSO&SGLASSO&CLIME&SCLIME&GLASSO&SGLASSO&CLIME&SCLIME&GLASSO&SGLASSO\\ \hline
 \multicolumn{13}{c}{Frobenius norm}\\ \hline
30&3.24(0.25)&3.37(0.15)&3.50(0.21)&3.54(0.19)&4.26(1.08)&3.38(0.16)&4.37(0.35)&3.55(0.2)&4.18(0.37)&3.4(0.15)&4.31(0.27)&3.54(0.22)\\
60&5.51(0.19)&5.44(0.13)&5.8(0.15)  &5.82(0.13)& 6.92(1.07)&5.45(0.12)&6.77(0.36)  &5.84(0.14)&  6.8(0.36)&5.46(0.13)&6.74(0.23)&5.82(0.14)  \\
90& 7.23(0.17)&7.1(0.13)&7.52(0.13)  &7.54(0.08)& 8.33(2.04)&7.12(0.13)&8.55(0.34)  &7.54(0.08) &8.53(0.45)&7.11(0.13)&8.51(0.24) &7.54(0.08) \\
120 & 8.7(0.15)&8.53(0.12)&8.95(0.1)  &8.96(0.11)&  10.46(0.8)&8.55(0.12)&10.03(0.4)&8.95(0.11)  &  10.2(0.43)&8.54(0.1)&10.02(0.23) &8.97(0.12) \\ \hline
 \multicolumn{13}{c}{Matrix $l_1$-norm}\\ \hline
30&2.66(0.17)&2.82(0.17)&2.85(0.15)&2.87(0.14)&3.12(0.69)&2.78(0.15)&3.37(0.22) & 2.86(0.12)&3.21(0.17)&2.81(0.14)&3.28(0.16)&2.86(0.14) \\
60&3.07(0.12)&3.19(0.15)&3.37(0.14) &3.41(0.13) & 3.53(0.41)&3.17(0.14)&3.92(0.29) &3.37(0.12)& 3.52(0.13)&3.16(0.15)&3.79(0.22)&3.41(0.12)   \\
90& 3.21(0.12)&3.28(0.14)&3.6(0.15) &3.58(0.11)&3.47(0.83)&3.3(0.15)&4.19(0.31)  &3.6(0.13)& 3.67(0.15)&3.29(0.16)&4.08(0.26)&3.58(0.1)  \\
120 &  3.3(0.13)&3.35(0.13)&3.7(0.16) &3.72(0.17) & 3.69(0.17)&3.33(0.11)&4.42(0.32)  &3.72(0.16)& 3.72(0.16)&3.32(0.12)&4.31(0.29) &3.7(0.16)\\ \hline
   \multicolumn{13}{c}{Operator norm}\\ \hline
30&1.84(0.16)&1.88(0.10)&1.95(0.13)&1.98(0.11)&2.28(0.54)&1.9(0.11)&2.31(0.2)&1.99(0.12)&2.25(0.21)&1.91(0.1)&2.29(0.18)&1.98(0.13)\\
60&2.25(0.08)&2.2(0.07)&2.3(0.07) &2.31(0.06) & 2.63(0.35)&2.21(0.05)&2.56(0.12)  &2.32(0.06) &2.62(0.12)&2.22(0.06)&2.56(0.1) &2.31(0.07) \\
90& 2.39(0.05)&2.35(0.05)&2.43(0.05) &2.44(0.03) & 2.61(0.62)&2.35(0.05)&2.65(0.09) & 2.44(0.04)&  2.69(0.11)&2.35(0.04)&2.64(0.08) &2.43(0.03) \\
120 & 2.48(0.04)&2.43(0.04)&2.5(0.04) &2.5(0.04) &  2.81(0.14)&2.44(0.04)&2.69(0.09) & 2.5(0.04)& 2.78(0.09)&2.44(0.03)&2.69(0.07) &2.5(0.04)\\ \hline \hline
	\end{tabular}}
	\label{tab1}
\end{table}

\begin{table}[htbp]
\caption{Comparisons of average (standard deviation) matrix losses for three distributions under Model II over 100 replications}
	\centering
 {\tiny 
 \setlength{\tabcolsep}{3pt}
 \renewcommand{\arraystretch}{2}
	\begin{tabular}{l|cccc|cccc|cccc} \hline
 &\multicolumn{4}{c}{Normal Distribution} &\multicolumn{4}{c}{$t_3$ Distribution} &\multicolumn{4}{c}{Mixture Normal Distribution}\\ \hline
 $p$&CLIME&SCLIME&GLASSO&SGLASSO&CLIME&SCLIME&GLASSO&SGLASSO&CLIME&SCLIME&GLASSO&SGLASSO\\ \hline
 \multicolumn{13}{c}{Frobenius norm}\\ \hline
30&1.61(0.13)&1.65(0.17)&1.53(0.1)&1.75(0.13)&2.39(0.47)&1.84(0.14)&2.24(0.26)&1.77(0.11)&2.28(0.16)&1.84(0.15)&2.14(0.15)&1.74(0.13)\\
60&3.34(0.17)&3.18(0.12)&2.85(0.08)&2.87(0.09)&3.97(0.52)&3.18(0.13)&3.56(0.2)&2.84(0.1)&4.01(0.17)&3.2(0.13)&3.56(0.12) &2.86(0.09) \\
90& 4.37(0.11)&4.24(0.08)&3.87(0.07) &4.05(0.07) &5.22(0.47)&4.25(0.08)&4.63(0.17)&4.06(0.08) &  5.21(0.15)&4.24(0.08)&4.65(0.1)&4.06(0.07) \\
120 & 5.58(0.1)&5.36(0.07)&4.88(0.07)&4.93(0.06) &  6.16(1.31)&5.36(0.07)&5.58(0.13) & 4.94(0.07)&  6.39(0.16)&5.35(0.07)&5.67(0.07)&4.94(0.07) \\ \hline
 \multicolumn{13}{c}{Matrix $l_1$-norm}\\ \hline
30&1.28(0.2)&1.39(0.29)&1.37(0.19)&1.64(0.21)&1.85(0.35)&1.67(0.23)&2.01(0.35)&1.62(0.18)&1.75(0.21)&1.67(0.25)&1.87(0.24) &1.63(0.21)\\
60&2.43(0.23)&2.42(0.23)&2.21(0.18)&2.38(0.21)&2.73(0.36)&2.45(0.24)&2.77(0.29)&2.37(0.2)&2.79(0.17)&2.46(0.22)&2.74(0.22)&2.38(0.21)\\
90&2.49(0.16)&2.6(0.18)&2.59(0.16)&2.85(0.19)& 2.9(0.2)&2.65(0.2)&3.09(0.29) &2.83(0.18)&2.91(0.16)&2.62(0.21)&3.05(0.22)&2.86(0.17)\\
120 &  3.3(0.21)&3.2(0.2)&3.19(0.19) &3.21(0.19) & 3.53(0.75)&3.2(0.2)&3.66(0.28)&3.26(0.2)  &  3.68(0.18)&3.2(0.2)&3.68(0.22)&3.23(0.19)  \\ \hline
   \multicolumn{13}{c}{Operator norm}\\ \hline
30&0.73(0.12)&0.75(0.09)&0.7(0.1)&0.77(0.08)&1.04(0.22)&0.74(0.08)&0.97(0.12)&0.78(0.09)&0.99(0.13)&0.74(0.08)&0.92(0.1)&0.78(0.08)\\
60&1.22(0.1)&1.01(0.06)&0.95(0.05)&0.99(0.06)&1.41(0.2)&1.02(0.06)&1.2(0.08)&1(0.06)&1.42(0.09)&1.01(0.05)&1.19(0.07)&1(0.05)  \\
90& 1.35(0.06)&1.25(0.04)&1.14(0.03) &1.27(0.04) & 1.63(0.16)&1.26(0.05)&1.41(0.07) &1.27(0.05) &  1.62(0.08)&1.25(0.04)&1.4(0.06) &1.27(0.04) \\
120 & 1.75(0.06)&1.63(0.04)&1.45(0.03) &1.47(0.04) & 1.94(0.41)&1.63(0.04)&1.7(0.06) &1.47(0.03) &  1.99(0.08)&1.63(0.04)&1.73(0.05)&1.47(0.03)
\\ \hline \hline
	\end{tabular}}
	\label{tab2}
\end{table}

\begin{table}[htbp]
\caption{Comparison of average (standard deviation) matrix losses for three distributions under Model III over 100 replications} 
\centering
{\tiny
	
 \setlength{\tabcolsep}{3pt}
 \renewcommand{\arraystretch}{2}
	\begin{tabular}{l|cccc|cccc|cccc} \hline
 &\multicolumn{4}{c}{Normal Distribution} &\multicolumn{4}{c}{$t_3$ Distribution} &\multicolumn{4}{c}{Mixture Normal Distribution}\\ \hline
 $p$&\tiny CLIME&\tiny SCLIME&\tiny GLASSO&\tiny SGLASSO&\tiny CLIME&\tiny SCLIME&\tiny GLASSO&\tiny SGLASSO&\tiny CLIME&\tiny SCLIME&\tiny GLASSO&\tiny SGLASSO\\ \hline
 \multicolumn{13}{c}{Frobenius norm}\\ \hline
30&1.5(0.14)&1.59(0.15)&1.36(0.11)&1.45(0.13)&2.02(0.32)&1.72(0.15)&1.98(0.27)&1.45(0.12)&2.02(0.19)&1.7(0.14)&1.94(0.19)&1.44(0.12)\\
60&2.16(0.14)&2.27(0.14)&2.04(0.1)&2.09(0.11)&2.77(0.36)&2.24(0.13)&2.9(0.27)&2.09(0.1)&2.83(0.16)&2.25(0.13)&2.87(0.15) &2.1(0.11) \\
90&2.67(0.12)&2.73(0.12)&2.6(0.11)&2.64(0.11)&3.5(0.54)&2.72(0.11)&3.58(0.24)&2.63(0.09)&3.45(0.19)&2.73(0.11)&3.59(0.13)&2.64(0.12)\\
120 & 3.1(0.11)&3.16(0.1)&3.08(0.1) &3.11(0.11) &  4.06(0.71)&3.13(0.11)&4.23(0.24)&3.1(0.11)  & 4.01(0.22)&3.16(0.11)&4.26(0.13) &3.11(0.1) \\ \hline
 \multicolumn{13}{c}{Matrix $l_1$-norm}\\ \hline
30&1.2(0.25)&1.43(0.26)&1.17(0.16)&1.12(0.19)&1.6(0.39)&1.46(0.25)&1.64(0.4)&1.14(0.19)&1.52(0.39)&1.44(0.24)&1.5(0.28) &1.11(0.17)\\
60&1.23(0.21)&1.35(0.23)&1.3(0.17)&1.29(0.15)&1.57(0.4)&1.32(0.19)&1.96(0.4)&1.3(0.16)&1.52(0.26)&1.31(0.21)&1.78(0.26)&1.31(0.14)\\
90&1.22(0.17)&1.26(0.16)&1.42(0.13)&1.42(0.16)&1.56(0.28)&1.24(0.15)&2.13(0.38)&1.42(0.16)&1.48(0.23)&1.25(0.14)&1.92(0.26&1.43(0.17))\\
120 & 1.18(0.13)&1.18(0.12)&1.52(0.12) &1.51(0.13) & 1.51(0.29)&1.19(0.16)&2.23(0.36) &1.52(0.13) &  1.47(0.19)&1.17(0.12)&2.1(0.29) &1.51(0.12) \\ \hline
   \multicolumn{13}{c}{Operator norm}\\ \hline
30&0.57(0.07)&0.62(0.06)&0.51(0.06)&0.54(0.06)&0.78(0.14)&0.63(0.06)&0.79(0.15)&0.54(0.07)&0.75(0.08)&0.62(0.06)&0.71(0.09)&0.54(0.06)\\
60&0.6(0.05)&0.62(0.05)&0.56(0.05)&0.57(0.05)&0.76(0.11)&0.61(0.05)&0.84(0.12)&0.56(0.04)&0.77(0.06)&0.61(0.05)&0.76(0.07) &0.57(0.04) \\
90&0.61(0.05)&0.62(0.05)&0.57(0.04)&0.58(0.04)&0.79(0.11)&0.61(0.04)&0.86(0.1)&0.58(0.03)&0.79(0.07)&0.62(0.05)&0.79(0.05)&0.58(0.04)
   \\
120 & 0.63(0.06)&0.65(0.06)&0.59(0.04) &0.6(0.04) & 0.79(0.17)&0.63(0.06)&0.88(0.09)&0.59(0.03)  &  0.8(0.06)&0.64(0.06)&0.81(0.05)&0.6(0.03) \\ \hline \hline
	\end{tabular}}
	\label{tab3}
\end{table}

\subsubsection{Graphical Model Recovery}
For simulation study, we adopt the same setting as in \cite{10.1214/12-AOS1037}. To generate a $p$-dimensional sparse graph $G = (V, E)$ where $V = \{1,\cdots, p \} $ correspond to variables ${\bm X} = (X_1, . . . , X_p)\in \R^p$, we associate each index $j \in {1, . . . , p}$ with a bivariate data point $(Y^{(1)}_j , Y^{(2)}_j ) \in [0, 1]^2$ where $Y^{(k)}_1 , \cdots, Y^{(k)}_p\sim \Unif[0, 1]$ for $k = 1, 2$. Each pair of vertices $(i, j)$ is included in the edge set $E$ with probability $\P((i, j) \in E) = \exp(-\|y_i-y_j\|_2^2/0.25)/\sqrt{2\pi}$, where $y_i \triangleq (y^{(1)}_i , y^{(2)}_i )$ is the empirical observation of $(Y^{(1)}_i , Y^{(2)}_i )$. We restrict the maximum degree of the graph to be 4 and and build the inverse correlation matrix $\bm{\Omega }$ according to $\bm{\Omega}_{jk} = 1$ if $j = k$, $\bm{\Omega}_{jk} = 0.145$ if $(j, k) \in E$, and $\bm{\Omega}_{jk} = 0$ otherwise. Let $\bm{\Sigma}=\bm{\Omega}^{-1}$. We randomly sample $n$ data points from three different typical elliptical distributions with covariance matrix $\bm{\Sigma}$.
\begin{itemize}
    \item Multivariate normal distribution: $N(0,\bm{\Sigma})$, 
    \item Multivariate t-distribution: $t_3(0,\bm{\Sigma})/\sqrt{3}$, 
    \item Mixed multivariate normal distribution: $MN(0.8,3,\bm{\Sigma})/\sqrt{2.6}$.
\end{itemize}

We consider sample size $n=400$, dimension $p=100$. We repeatedly generate the data matrix $X$ for 100 times and compute the averaged True Positive Rates and False Positive Rates using a path of tuning parameters $\lambda$ from 0.005 to 1 and the tuning parameter $\tau_n$ is set to be $10^{-5}$ since slight difference can be found with different $\tau_n$. The recovery performance are evaluated by plotting the corresponding ROC curves, which are presented in Figure \ref{fig:schemes}.

Similarly, we observe that, in the case of multivariate normal distributions, the spatial-sign-based methods and classical methods perform similarly in terms of recovery. However, when dealing with heavy-tailed distributions, the spatial-sign-based methods significantly outperform the classical methods, demonstrating their superior robustness and ability to handle non-normal data distributions effectively. This highlights the advantages of incorporating spatial-sign covariance matrices in high-dimensional estimation problems, particularly in the presence of heavy-tailed distribution.
\begin{figure}[htbp]
    \caption{ Comparisons of different schemes. Each subfigure represents the ROC curves for a specific scheme.}
    \centering
    \begin{subfigure}[b]{0.45\textwidth}
        \includegraphics[width=\textwidth]{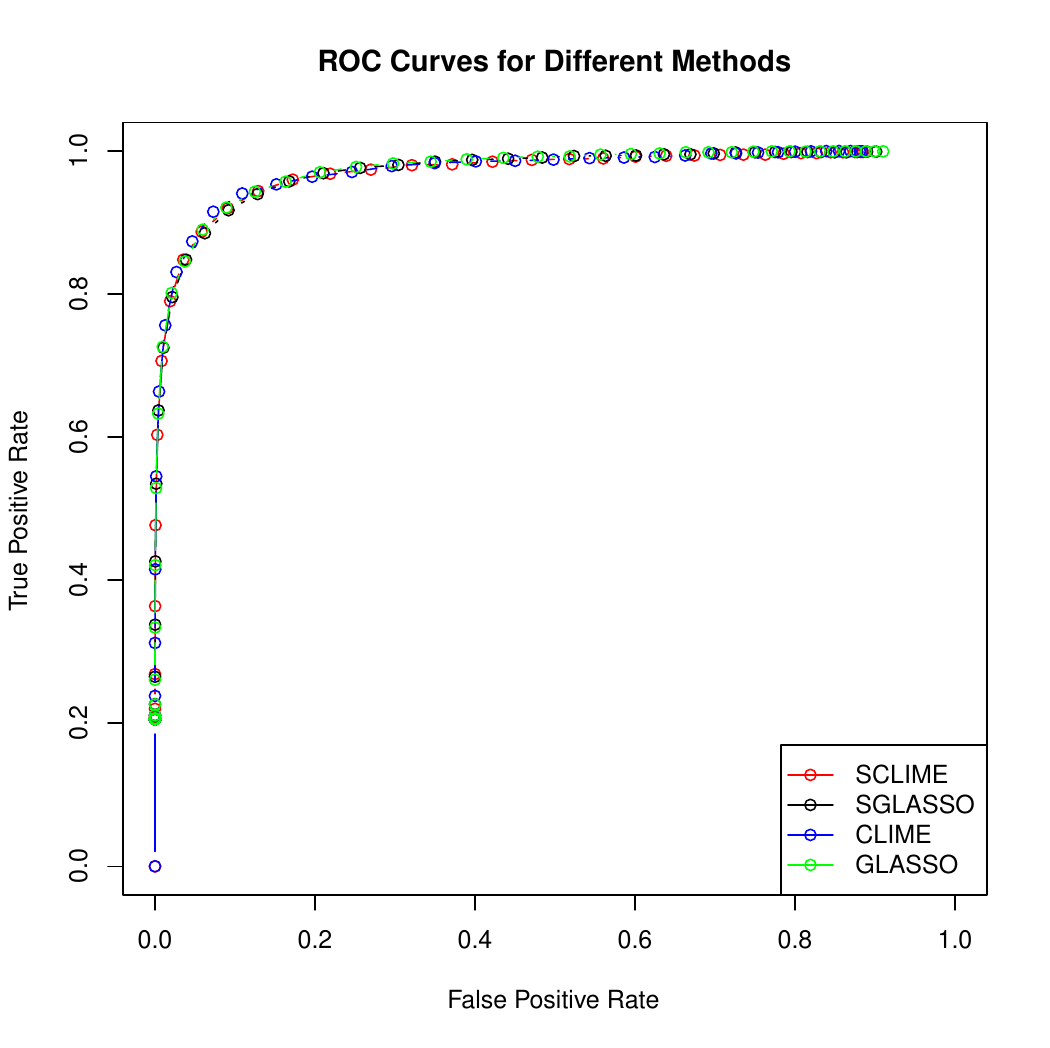}
        \caption{Normal Distribution}
        \label{fig:n}
    \end{subfigure}
    \begin{subfigure}[b]{0.45\textwidth}
        \includegraphics[width=\textwidth]{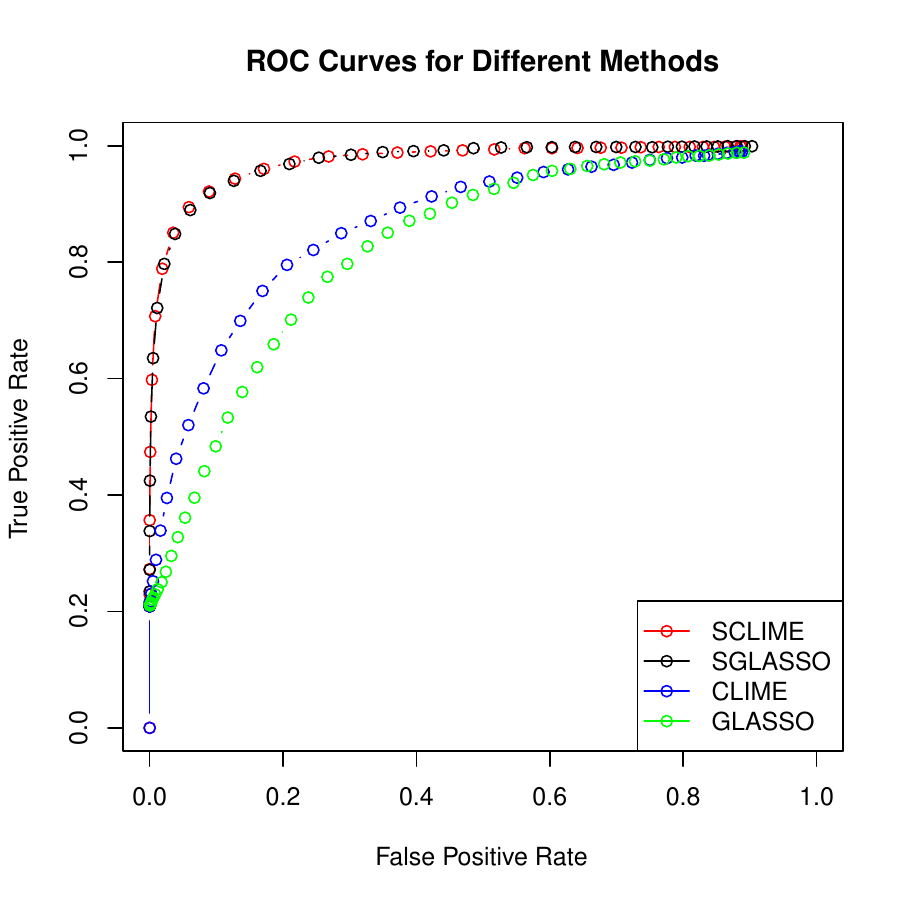}
        \caption{Multivariate $t_3$ Distribution}
        \label{fig:t}
    \end{subfigure}
    \begin{subfigure}[b]{0.45\textwidth}
        \includegraphics[width=\textwidth]{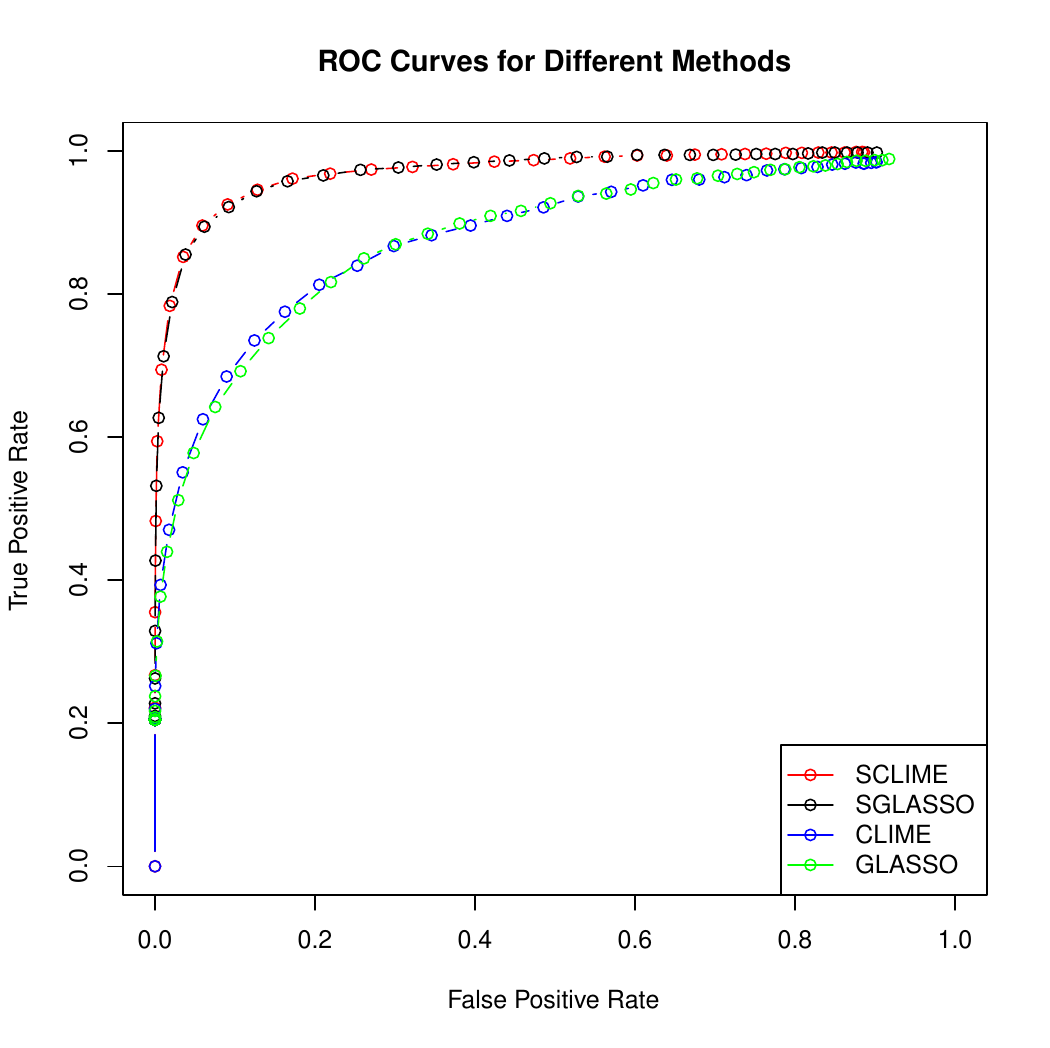}
        \caption{Mixture Normal Distribution}
        \label{fig:m}
    \end{subfigure}
    \label{fig:schemes}
\end{figure}
\subsubsection{Linear Discriminant Analysis}
For simulation study, we follow the classical data generating process of Linear Discriminant Analysis. For each data point pair $({\bm X},Y)$, we first generate $Y$ from $\text{Bernoulli}(1,p_1)$, and for $Y=k, k=0,1$, we generate ${\bm X}$ from ${\bm X}|Y=k\sim EC_p({\bm \mu}_k, {\bf\Sigma}, r^*), k =0, 1$. Specifically, for mean vectors, without loss of generality, we set ${\bm \mu}_0=0$ and ${\bm\mu}_1=[a\bm{1}_s,\0_{p-s}]$, where $s$ is the parameter for sparsity control and $a$ represents the distance between two different distributions, aligning with the assumption for ${\bm\mu}_d$ mentioned in Theorem \ref{LDATheorem} and \ref{LDATheoremg}. For covariance matrix, we consider Model I and II used in the previous simulation, see section \ref{simu: pme}, to set $\bf{\Sigma}=\bf{\Omega}^{-1}$. Similarly, we consider three classical elliptical distributions with  different  types of $r^*$ as before.
\begin{itemize}
    \item Multivariate normal distribution: $N(0,\bm{\Sigma})$, 
    \item Multivariate t-distribution: $t_3(0,\bm{\Sigma})/\sqrt{3}$, 
    \item Mixed multivariate normal distribution: $MN(0.8,3,\bm{\Sigma})/\sqrt{2.6}$.
\end{itemize}

We first randomly sample $n$ data points as training data and another independent sample of size $n$ as validation data. The tuning parameter $\lambda_n$ is selected based on misclassification rate. To be more specific, we compute a series of estimators with 50 different values of $\lambda_n$ and use the one with the smallest misclassification rate, which is defined by 
\[
    \text{Misclassificaton Rate}=\frac{\#\{i: \hat Y_i\neq Y_i\}}{n}
\]
 where $\hat Y_i =\1\{({\bm X}_i-\hat{\bm\mu}^k_a)^\sT \hat {\bf V}_{k}\hat {\bm \mu}^k_d>0\}, k=\text{CLIME, SCLIME, GLASSO, SGLASSO}$. Notably, for $\hat{\bm \mu}^{k}_d$ and $\hat{\bm\mu}^{k}_a$, we use spatial median estimator as estimators for spatial-sign based methods and naive mean estimator for other classical methods. Finally, we independently generate a test sample of sample size $n$ to test the performance of those estimators.
 
 We consider sample size $n=200$ and $p_1=0.5$ for four different dimensions $p=30, 60, 90, 120$.  For all dimensions, we set the sparsity level $s=10$ and distance level $a=0.05$. All the results are based on 100 replications. The classification performance is evaluated by specificity, sensitivity, and Mathews correlation coefficient (MCC) criteria, defined as follows.
 \[
\text{Specificity} = \frac{\text{TN}}{\text{TN} + \text{FP}}, \quad
\text{Sensitivity} = \frac{\text{TP}}{\text{TP} + \text{FN}},
\]

\[
\text{MCC} = \frac{\text{TP} \times \text{TN} - \text{FP} \times \text{FN}}
{\sqrt{(\text{TP} + \text{FP})(\text{TP} + \text{FN})(\text{TN} + \text{FP})(\text{TN} + \text{FN})}}.
\]
Here TP and TN stand for true positives ($Y=1$) and true negatives ($Y=0$), respectively, and FP and FN stand for false positives/negatives. The larger the criterion value, the better the classification performance. The averages and standard errors of the foregoing criteria are reported in Table \ref{tab4}. 

Under the multivariate normal distribution, spatial-sign-based methods achieve similar performance as classical methods, with no significant differences observed. However, when applied to non-normal distributions, our proposed SCLIME and SGLASSO methods consistently outperform the corresponding classical methods, particularly in terms of the MCC index. This indicates that our methods are more effective at handling the complexities of non-normal data. Therefore, our proposed methods are highly preferable in practical applications, especially when the underlying distribution exhibits heavy-tailed properties, where traditional methods may struggle. The robustness of SCLIME and SGLASSO makes them a valuable tool for high-dimensional statistical inference in the presence of heavy-tailed distribution.

\begin{table}[htbp]
\caption{ Comparisons of average (standard deviation) classification errors for different distributions under Model I}
\centering
{\tiny
 \setlength{\tabcolsep}{3pt}
 \renewcommand{\arraystretch}{2}
	\begin{tabular}{l|cccc|cccc|cccc} \hline
 &\multicolumn{4}{c}{Normal Distribution} &\multicolumn{4}{c}{$t_3$ Distribution} &\multicolumn{4}{c}{Mixture Normal Distribution}\\ \hline
$p$&CLIME&SCLIME&GLASSO&SGLASSO&CLIME&SCLIME&GLASSO&SGLASSO&CLIME&SCLIME&GLASSO&SGLASSO\\ \hline
 \multicolumn{13}{c}{Specificity}\\ \hline
30&0.892(0.04)&0.885(0.04)&0.898(0.04)&0.889(0.04)&0.773(0.10)&0.830(0.05)&0.821(0.05)&0.842(0.04)&0.868(0.06)&0.854(0.05)&0.872(0.05)&0.864(0.04)\\
60&0.868(0.04)&0.856(0.04)&0.866(0.05)&0.862(0.04)&0.741(0.09)&0.809(0.05)&0.798(0.06)&0.817(0.05)&0.823(0.07)&0.826(0.05)&0.813(0.06)&0.835(0.05)\\
90&0.846(0.01)&0.834(0.01)&0.845(0.01)&0.843(0.01)&0.686(0.01)&0.777(0.01)&0.759(0.01)&0.783(0.01)&0.804(0.07)&0.808(0.05)&0.781(0.06)&0.808(0.05)\\
120&0.817(0.06)&0.8(0.06)&0.804(0.06)&0.799(0.06)&0.675(0.08)&0.766(0.03)&0.755(0.05)&0.779(0.03)&0.786(0.04)&0.77(0.05)&0.763(0.03)&0.773(0.04)\\ \hline
 \multicolumn{13}{c}{Sensitivity}\\ \hline
30&0.889(0.04)&0.885(0.04)&0.900(0.04)&0.893(0.03)&0.780(0.13)&0.836(0.05)&0.83(0.05)&0.843(0.04)&0.807(0.07)&0.866(0.04)&0.823(0.06)&0.875(0.04)\\
60&0.864(0.04)&0.852(0.04)&0.866(0.04)&0.859(0.04)&0.731(0.12)&0.807(0.05)&0.795(0.05)&0.817(0.04)&0.782(0.07)&0.841(0.05)&0.796(0.06)&0.847(0.05)\\
90&0.842(0.01)&0.831(0.01)&0.843(0.01)&0.840(0.01)&0.704(0.02)&0.792(0.01)&0.781(0.01)&0.805(0.01)&0.752(0.08)&0.809(0.06)&0.763(0.07)&0.813(0.06)\\
120&0.822(0.05)&0.8(0.05)&0.822(0.04)&0.809(0.04)&0.693(0.07)&0.75(0.06)&0.737(0.05)&0.772(0.06)&0.707(0.03)&0.773(0.03)&0.735(0.03)&0.769(0.04)\\ \hline
 \multicolumn{13}{c}{MCC}\\ \hline
30&0.784(0.04)&0.772(0.05)&0.800(0.04)&0.783(0.04)&0.570(0.10)&0.669(0.05)&0.654(0.05)&0.687(0.05)&0.682(0.06)&0.723(0.05)&0.699(0.06)&0.741(0.05)\\
60&0.734(0.05)&0.710(0.05)&0.734(0.06)&0.722(0.06)&0.481(0.11)&0.619(0.06)&0.595(0.06)&0.636(0.06)&0.612(0.07)&0.671(0.06)&0.613(0.07)&0.685(0.05)\\
90&0.690(0.06)&0.666(0.07)&0.689(0.07)&0.683(0.07)&0.395(0.13)&0.571(0.07)&0.541(0.07)&0.589(0.07)&0.562(0.07)&0.620(0.06)&0.548(0.07)&0.623(0.06)\\
120 &0.64(0.06)&0.602(0.05)&0.628(0.04)&0.61(0.05)&0.37(0.09)&0.517(0.06)&0.493(0.05)&0.551(0.06)&0.495(0.04)&0.543(0.05)&0.498(0.05)&0.542(0.05)\\ \hline \hline
	\end{tabular}}
	\label{tab4}
\end{table}
\begin{table}[htbp]
\caption{ Comparisons of average (standard deviation)  classification errors for different distributions under Model II}
\centering
{\tiny
 \setlength{\tabcolsep}{3pt}
 \renewcommand{\arraystretch}{2}
	\begin{tabular}{l|cccc|cccc|cccc} \hline
 &\multicolumn{4}{c}{Normal Distribution} &\multicolumn{4}{c}{$t_3$ Distribution} &\multicolumn{4}{c}{Mixture Normal Distribution}\\ \hline
$p$&CLIME&SCLIME&GLASSO&SGLASSO&CLIME&SCLIME&GLASSO&SGLASSO&CLIME&SCLIME&GLASSO&SGLASSO\\ \hline
 \multicolumn{13}{c}{Specificity}\\ \hline
30 &0.771(0.06)&0.777(0.05)&0.781(0.05)&0.782(0.05)&0.694(0.12)&0.739(0.08)&0.721(0.08)&0.744(0.07)&0.778(0.08)&0.763(0.07)&0.791(0.07)&0.767(0.06)\\
60 &0.684(0.08)&0.734(0.05)&0.746(0.06)&0.741(0.05)&0.623(0.11)&0.701(0.07)&0.685(0.07)&0.708(0.07)&0.698(0.09)&0.728(0.07)&0.733(0.08)&0.737(0.07)\\
90 &0.635(0.07)&0.698(0.06)&0.711(0.06)&0.707(0.06)&0.604(0.14)&0.687(0.07)&0.673(0.07)&0.693(0.08)&0.637(0.09)&0.691(0.06)&0.695(0.07)&0.700(0.06)\\
120&0.619(0.02)&0.686(0.01)&0.674(0.01)&0.692(0.01)&0.570(0.11)&0.651(0.06)&0.638(0.07)&0.662(0.06)&0.594(0.01)&0.657(0.01)&0.676(0.01)&0.698(0.01)\\ \hline
 \multicolumn{13}{c}{Sensitivity}\\ \hline
30&0.760(0.06)&0.770(0.05)&0.780(0.05)&0.776(0.05)&0.669(0.15)&0.743(0.08)&0.726(0.07)&0.748(0.07)&0.692(0.09)&0.761(0.07)&0.715(0.08)&0.770(0.07)\\
60&0.692(0.07)&0.735(0.06)&0.746(0.06)&0.743(0.06)&0.624(0.12)&0.710(0.07)&0.692(0.07)&0.713(0.06)&0.633(0.09)&0.720(0.07)&0.686(0.08)&0.725(0.07)\\
90&0.651(0.07)&0.702(0.06)&0.711(0.06)&0.709(0.06)&0.571(0.14)&0.672(0.06)&0.666(0.06)&0.680(0.07)&0.612(0.09)&0.698(0.06)&0.680(0.07)&0.713(0.07)\\
120&0.59(0.02)&0.669(0.01)&0.681(0.01)&0.679(0.01)&0.576(0.12)&0.655(0.06)&0.650(0.07)&0.664(0.07)&0.624(0.01)&0.696(0.01)&0.676(0.01)&0.718(0.01)\\ \hline
 \multicolumn{13}{c}{MCC}\\ \hline
30&0.532(0.08)&0.547(0.07)&0.563(0.07)&0.559(0.07)&0.377(0.11)&0.485(0.08)&0.451(0.09)&0.495(0.08)&0.477(0.07)&0.526(0.06)&0.511(0.07)&0.539(0.07)\\
60&0.377(0.10)&0.471(0.08)&0.494(0.07)&0.485(0.07)&0.252(0.10)&0.412(0.07)&0.379(0.08)&0.422(0.07)&0.335(0.09)&0.450(0.07)&0.422(0.08)&0.463(0.07)\\
90&0.286(0.09)&0.401(0.07)&0.424(0.07)&0.417(0.07)&0.179(0.09)&0.361(0.08)&0.341(0.07)&0.374(0.08)&0.251(0.11)&0.391(0.07)&0.377(0.07)&0.414(0.07)\\
120&0.214(0.09)&0.356(0.08)&0.356(0.07)&0.371(0.07)&0.149(0.09)&0.308(0.07)&0.289(0.07)&0.327(0.07)&0.219(0.09)&0.364(0.07)&0.353(0.07)&0.396(0.07)\\ \hline \hline
	\end{tabular}}
	\label{tab5}
\end{table}
\subsection{Real Data Analysis}
\subsubsection{Analysis of Equities Data}
In this section, we apply four methodologies, i.e. CLIME, SCLIME, GLASSO, SGLASSO, to analyze graphical model of stock data, specifically the Standard $\& $ Poor's 500 (S$\& $P 500) index. We collect the daily closing prices for $p=470$ stocks that are consistently in the S$\&$P 500 index between January 1, 2015 through December 31, 2019 from Yahoo! Finance through \texttt{yfinance} in \texttt{Python}. This gives us $n=1256$ data points in total, each of which corresponding to the vector of closing prices on a trading day. Denote the daily closing price corresponding to day $t$ and stock $j$ by $P_{t,j}$, we consider build the graph over indices $j$ based on daily return $X_{t,j}=\log(P_{t,j}/P_{t-1,j})$. 

To further compare the robustness of different methods, besides processing clean data, we additionally add some random noise to the data and process it to further validate the robustness of spatial-sign-based methods. Specifically, let $r \in [0, 1)$ represents the proportion of samples being contaminated. For each stock, we randomly select $\lceil nr \rceil$ entries and replace them with either $a$ or $-a$ with equal probability, where $a$ represents some extreme values that can viewed as outliers. We choose $a =0.13$ and $r=0.005$ for contamination. All the inputs of different methods go through covariance to correlation transformation. The tuning parameter is automatically selected using a stability based approach named StARS proposed in \cite{liu2010stability}. 
\begin{figure}[htbp]
    \caption{ Comparison of different methods for clean stock graphical model. Each subfigure represents the model  for a specific method.}
    \centering
    \begin{subfigure}[b]{0.45\textwidth}
        \includegraphics[width=\textwidth]{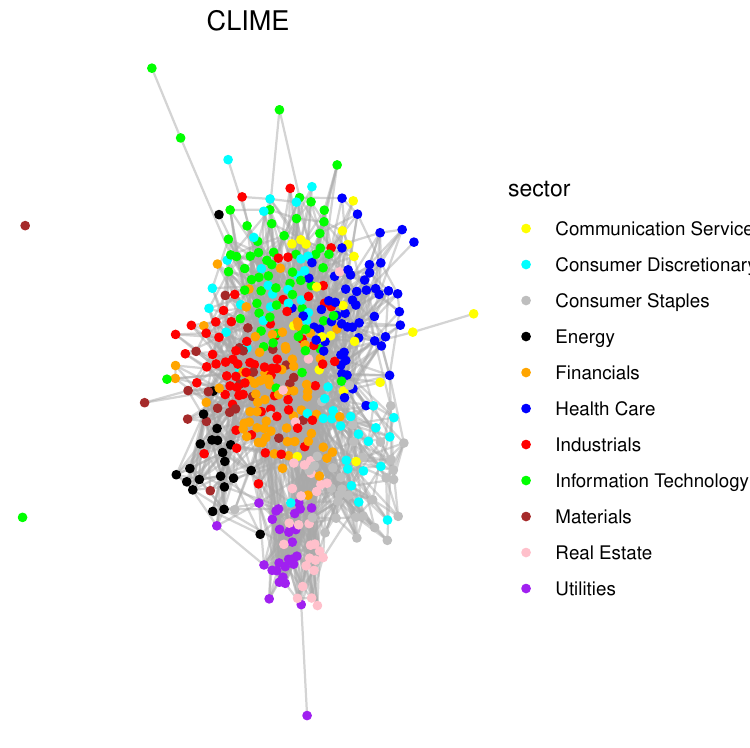}
        \caption{CLIME}
        
    \end{subfigure}
    \begin{subfigure}[b]{0.45\textwidth}
        \includegraphics[width=\textwidth]{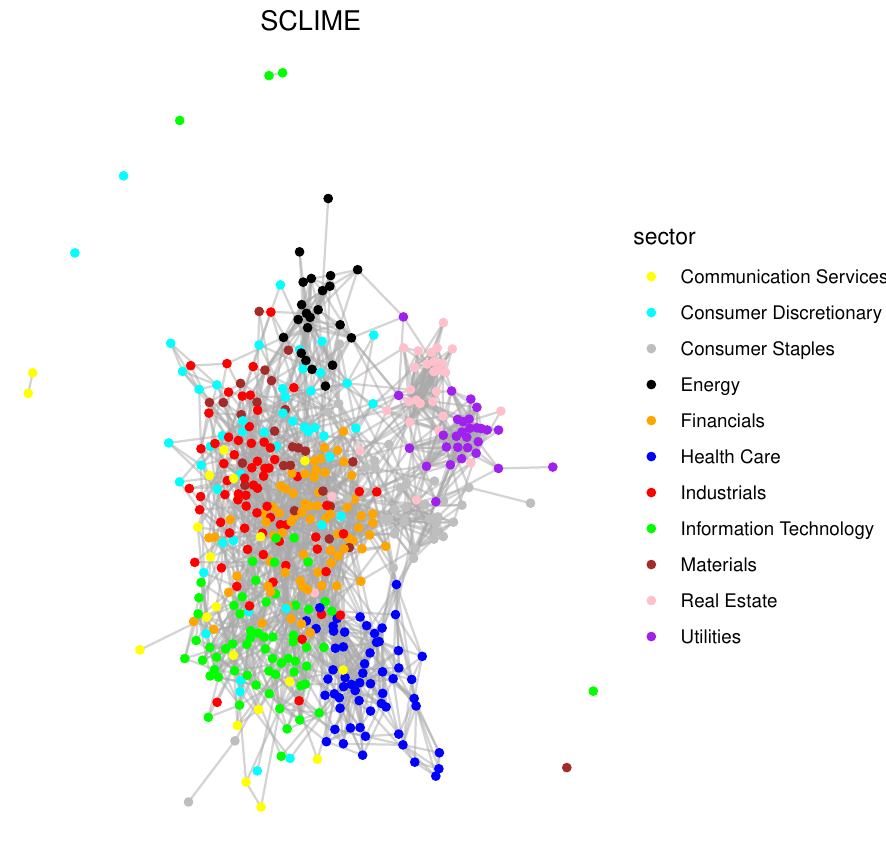}
        \caption{SCLIME}

    \end{subfigure}
    \begin{subfigure}[b]{0.45\textwidth}
        \includegraphics[width=\textwidth]{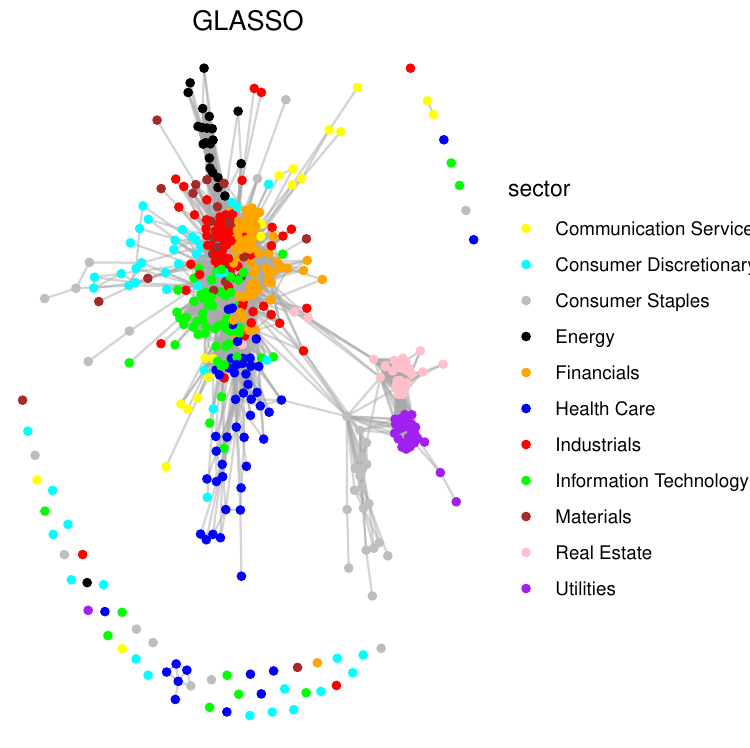}
        \caption{GLASSO}
    \end{subfigure}
    \begin{subfigure}[b]{0.45\textwidth}
        \includegraphics[width=\textwidth]{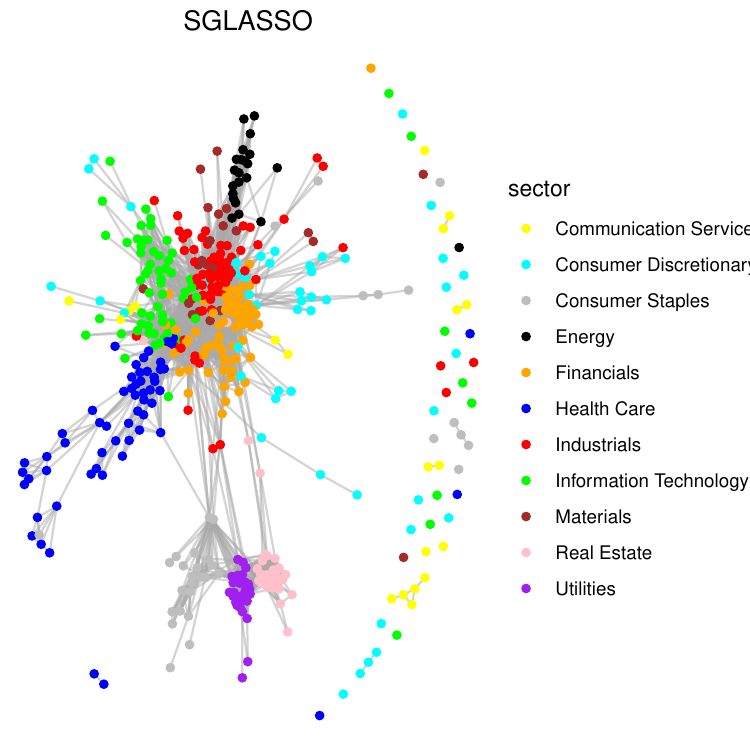}
        \caption{SGLASSO}
    \end{subfigure}
    \label{fig:gmstock}
\end{figure}
\begin{figure}[htbp]
    \caption{Comparisons of different methods for contaminated stock graphical model. Each subfigure represents the model  for a specific method.}
    \centering
    \begin{subfigure}[b]{0.45\textwidth}
        \includegraphics[width=\textwidth]{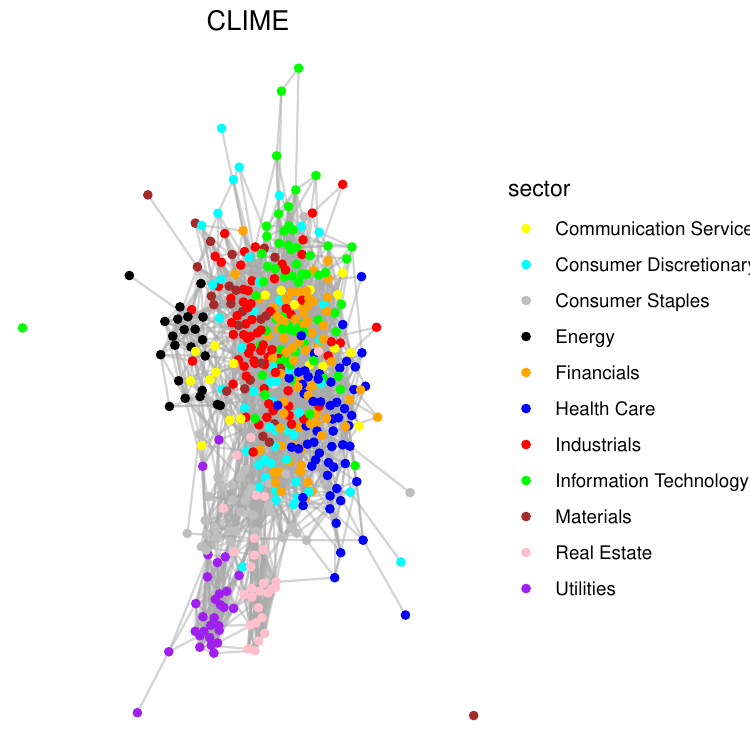}
        \caption{CLIME}
        
    \end{subfigure}
    \begin{subfigure}[b]{0.45\textwidth}
        \includegraphics[width=\textwidth]{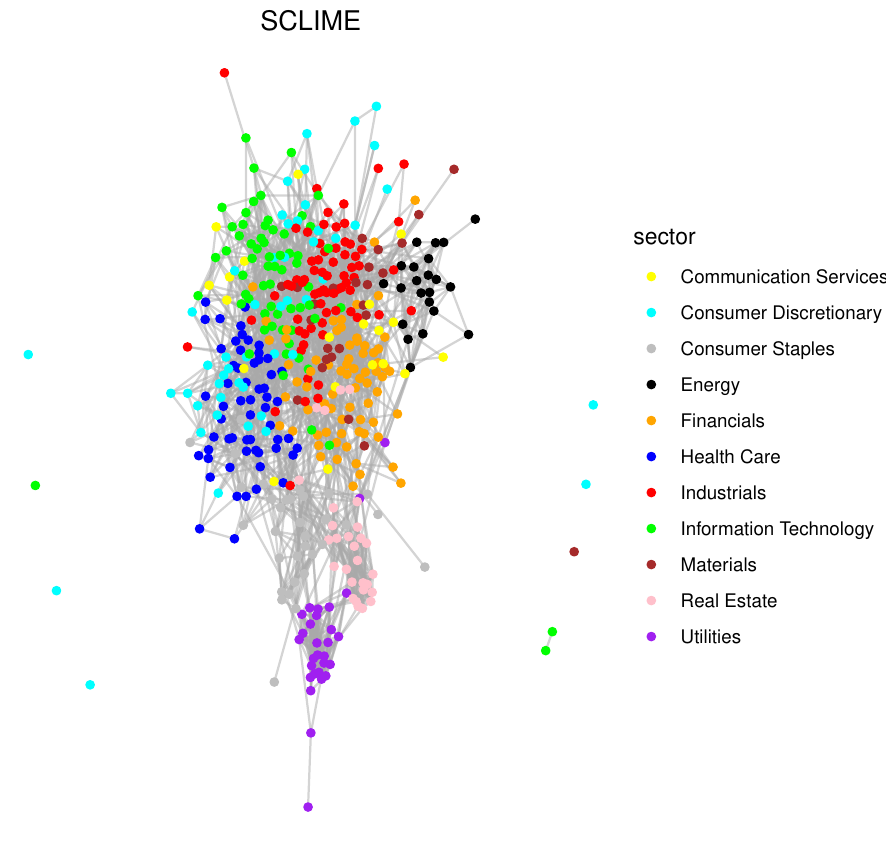}
        \caption{SCLIME}

    \end{subfigure}
    \begin{subfigure}[b]{0.45\textwidth}
        \includegraphics[width=\textwidth]{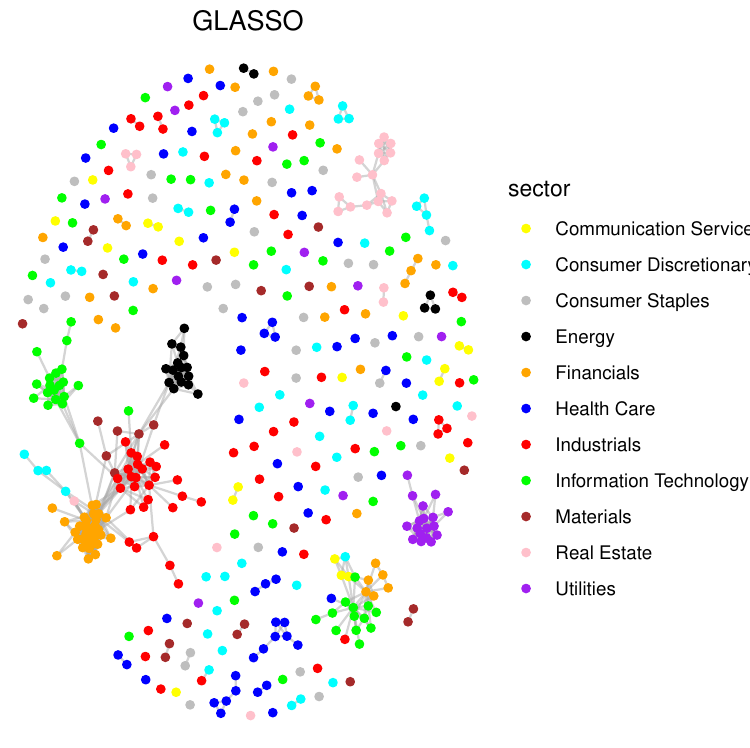}
        \caption{GLASSO}
    \end{subfigure}
    \begin{subfigure}[b]{0.45\textwidth}
        \includegraphics[width=\textwidth]{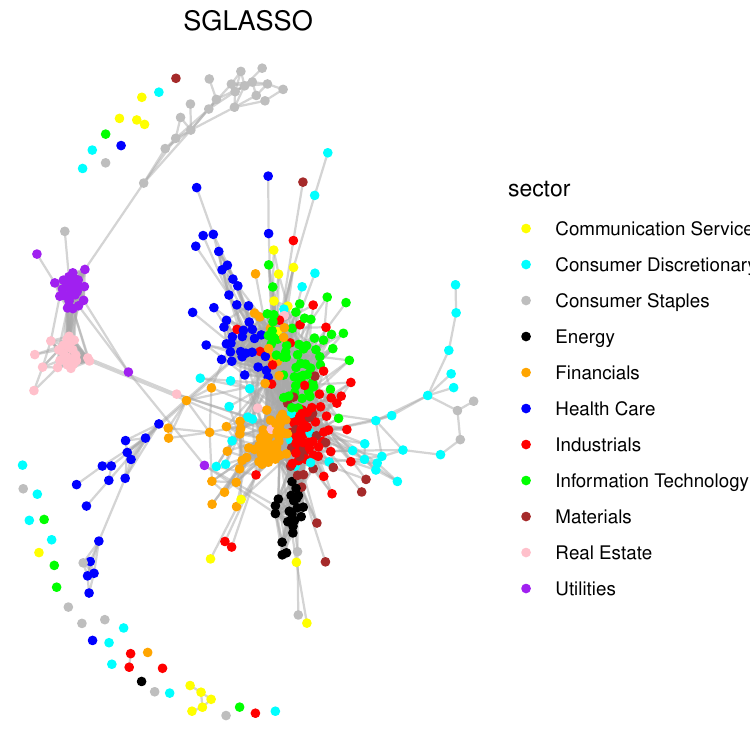}
        \caption{SGLASSO}
    \end{subfigure}
    \label{fig:gmstock_c}
\end{figure}

Figure \ref{fig:gmstock} and \ref{fig:gmstock_c} illustrate the estimated clean and contaminated stock graphical models, respectively, using the same layout setting, where the nodes are colored according to the GICS sector of the corresponding stock. The 470 stocks are categorized into 11 Global Industry Classification Standard (GICS) sectors, including \texttt{Communication Services} (20 stocks), \texttt{Consumer Discretionary} (49 stocks), \texttt{Consumer Staples} (34 stocks), \texttt{Energy} (22 stocks), \texttt{Financials} (71 stocks), \texttt{Health Care} (58 stocks), \texttt{Industrials} (69 stocks), \texttt{Information Technology} (64 stocks) \texttt{Telecommunications Services} (25 stocks), \texttt{Materials} (29 stocks), and \texttt{Utilities} (29 stocks). Besides, we compute the difference between estimated graphical models with clean data and those with contaminated data by the difference of edges. The result is illustrated in Table \ref{tab 6}.
\begin{table}[htbp]
    \centering
    \caption{ Comparisons of different edges   between clean and contaminated graphical models. }
    \begin{tabular}{lcccc}
        \hline
        CLIME & SCLIME & GLASSO &SGLASSO  \\
        \hline
        5712 & 3678 & 7418 & 4390 \\
        \hline
    \end{tabular}
    \label{tab 6}
\end{table}

Obviously, the performance can be evaluated as for clustering problems, for which we expect the nodes of the same color cluster together tightly to form a group  while different groups are separated clearly. The result exhibits that in terms of clustering, while SCLIME and CLIME achieve similar performance, SGLASSO outperforms GLASSO significantly, which demonstrates the robustness of spatial-sign methods. Besides, the difference in edges between clean and contaminated graphical models further verify the robustness of sptial-sign-based methods that exhibits approximately 40 \% less different edges. 
\subsubsection{Analysis of clear cell Renal Cell Carcinoma(ccRCC) Dataset}
In this section, we apply four methodologies, i.e. CLIME, SCLIME, GLASSO, SGLASSO, to analyze clear cell renal cell carcinoma (ccRCC) dataset in  \cite{von2014neuronal}. Clear cell renal cell carcinoma (ccRCC) is the most common subtype of kidney cancer and has the highest propensity to manifest as metastatic disease. The dataset consists of 54676 genes expression levels with 143 subjects, including 71 subjects identified as ccRCC and 72 normal subjects. The goal of our classification is to identify ccRCC from normal kidney tissue, so that we are more likely to learn biological insights of the disease. Based on the estimated inverse covariance matrix of the gene expression levels, we apply linear discriminant analysis (LDA) to predict whether or not a subject should be classified as ccRCC.

The preprocessing procedures are described as follows. First, we standardize each dimension of data matrix. Second, to reduce the computational cost and also maintain high dimensional setting, a two-sample t-test is performed between the two groups for each gene, and the 100 most significant genes (i.e., with the smallest p-values) are retained as the covariates for prediction. Third, to further compare the robustness of different methods, we additionally add some random noise to the data. Specifically, let $r \in [0, 1)$ represents the proportion of samples being contaminated. For each gene, we randomly select $\lceil nr \rceil$ entries and replace them with either $a$ or $-a$ with equal probability, where $a$ represents some extreme values that can viewed as outliers. We choose $r=0.1$ and $a=10$. 

And then following the LDA framework mentioned in section \ref{LDA}, we randomly select 50 samples from each group as the test set since the dataset is of nearly the same size, while the remaining forms the training set with the sample size 43 in total.  We use the estimated Fisher rule to classify test set and evaluate the performance based on  specificity, sensitivity, and Mathews correlation coefficient (MCC) criteria. For the tuning parameters, we select $\lambda$ based on 3-fold cross validation  on the training set. The results are based on 100 replications. 

\begin{table}[h]
    \centering
    \caption{ Comparisons of average (standard deviation) classification errors of ccRCC dataset over 100 replications. }
    \begin{tabular}{lcccc}
        \hline
        Method & Specificity & Sensitivity & MCC  \\
        \hline
        CLIME & 0.849(0.05) & 0.799(0.05) & 0.650(0.07)  \\
        SCLIME &  0.854(0.04) &  0.826(0.04) &  0.681(0.05)   \\
        GLASSO &  0.847(0.05) & 0.786(0.06) & 0.637(0.05) \\
        SGLASSO &  0.863(0.03) &  0.836(0.03) & 0.700(0.04) \\
        \hline
    \end{tabular}
    \label{tab 7}
\end{table}

As illustrated in Table \ref{tab 7}, the performance of spatial-sign-based methods significantly outperforms classical methods, especially in terms of the overall classification performance measured by MCC. This further demonstrates the robustness of spatial-sign-based methods for high-dimensional statistical inference.
\section{Conclusion}
This paper introduces two novel high-dimensional precision matrix estimators, termed SCLIME and SGLASSO, which replace the sample covariance matrix in the classical CLIME and GLASSO methods with the sample spatial-sign covariance matrix. Leveraging the robust properties of the spatial-sign covariance matrix, these new estimators exhibit enhanced performance, particularly under heavy-tailed distributions. We also establish the convergence rates of the proposed estimators, which align with existing theoretical results in the literature. To further demonstrate the utility of these methods, we apply SCLIME and SGLASSO to two statistical problems: elliptical graphical models and linear discriminant analysis. Both theoretical analysis and simulation studies substantiate the robustness and efficiency of the proposed methods, highlighting their superiority over traditional estimators in challenging settings.

Next, we propose potential directions for further research. Both SCLIME and SGLASSO utilize the \(l_1\)-norm penalization, which may introduce some bias in practice. To mitigate this, alternative penalization schemes, such as SCAD \citep{FanLi2001} or MCP \citep{Zhang2010}, could be considered. Additionally, following the approach of \cite{cai2011constrainedl1minimizationapproach}, a two-stage procedure could be employed to enhance the numerical performance of SCLIME and SGLASSO. This would involve refitting the estimator with the non-zero elements identified in the initial procedure. Furthermore, these robust high-dimensional precision matrix estimators could be extended to address a variety of other high-dimensional statistical problems, such as location parameter testing \citep{huang2022overview} and quadratic discriminant analysis \citep{qin2018review}.

\section{Appendix}
\subsection{Appendix A: SCLIME-related Proofs}
\begin{lemma}{(General moments of spherically symmetric distribution)}\label{mofs}

    Suppose ${\bm u}=[u_1,u_2,\cdots,u_p]^\sT\in\R^{p}$ is uniformly distributed on $\mathbb{S}^{p-1}$, then for any integers $m_1, \dots, m_p$, with $m = \sum_{i=1}^p m_i$, the mixed moments of $u$ can be expressed as:
\[
\mathbb{E} \left( \prod_{i=1}^p u_i^{m_i} \right) =
\begin{cases}
\frac{1}{(p/2)^{[l]}} \frac{(2l)!}{\prod_{i=1}^n 4^{l_i}(l_i)!}, & \text{if } m_i = 2l_i \text{ are even}, \, i = 1, \dots, p, \, m = 2l; \\
0, & \text{if at least one of the } m_i \text{ is odd.}
\end{cases}
\]
where as above $x^{[l]} = x(x+1)\cdots(x+l-1)$.

Specifically,
\[\E[u_iu_ju_ku_l]=\frac{\delta_{ij}\delta_{kl}+\delta_{ik}\delta_{jl}+\delta_{il}\delta_{jk}}{p(p+2)}\]
\end{lemma}
\begin{proof}
    This is Theorem 3.3 in Section 3.1.2. of \cite{fang2018symmetric}  
\end{proof}
\begin{lemma}\label{mofs2}
    Suppose ${\bm u}=[u_1,u_2,\cdots,u_p]^\sT\in\R^{p}$ is uniformly distributed on $\mathbb{S}^{p-1}$, then for $\ba_1, \ba_2\in\R^{p}$, 
    \[
    \E[\ba_1^\sT u \cdot \ba_2^\sT u]=\frac{1}{p}\ba_1^\sT \ba_2
    \]
    \[
    \E[(\ba_1^\sT u \cdot \ba_2^\sT u)^2]=\frac{\|\ba_1\|_2^2\|\ba_2\|_2^2+2(\ba_1^\sT\ba_2)^2}{p(p+2)}
    \]
\end{lemma}
\begin{proof}
    By Lemma \ref{mofs} and basic calculation.
\end{proof}
\begin{lemma}\label{L1ineq}
Suppose ${\bf A}\in \R^{m\times n}, {\bf B}\in \R^{n\times n}$ and ${\bf B}$ is symmetric, then
\[
\|{\bf AB}\|_\infty\leq \|{\bf A}\|_\infty\|{\bf B}\|_{L1}
\]
    
\end{lemma}

\begin{lemma}
    Suppose ${\bm u}=[u_1,\cdots, u_p]^\sT\in \R^p$ is uniformly distributed on $\mathbb{S}^{p-1}$, then under Assumption \ref{V0} and \ref{lambda}, we have 
    \[
    \left|\E[\frac{u_iu_j}{{\bm u}^\sT {\bf\Sigma}_0{\bm u}}]-\frac{1}{\tr({\bf \Sigma}_0)}\delta_{ij}\right|=O(p^{-\frac{3}{2}})
    \]
    Furthermore, we have 
    \[
    \left|\E[\frac{u_iu_j}{{\bm u}^\sT {\bf\Sigma}_0{\bm u}}]\right|=O(p^{-1})
    \]
    
\end{lemma}
\begin{proof}
By Lemma \ref{mofs},
    \begin{align*}
           \left|\E[\frac{u_iu_j}{{\bm u}^\sT {\bf\Sigma}_0{\bm u}}]-\frac{1}{\tr({\bf \Sigma}_0)}\delta_{ij}\right|
           =& \left|\E[\frac{u_iu_j}{{\bm u}^\sT {\bf \Sigma}_0{\bm u}}-\frac{u_iu_j}{\E[{\bm u}^\sT{\bf \Sigma}_0{\bm u}]}]\right|\\
           \leq &\E\left|\frac{u_iu_j}{{\bm u}^\sT {\bf \Sigma}_0{\bm u}}-\frac{u_iu_j}{\E[{\bm u}^\sT {\bf\Sigma}_0{\bm u}]}\right|\\
           =& \E\left|\frac{u_iu_j({\bm u}^\sT {\bf\Sigma}_0{\bm u}-\E[{\bm u}^\sT {\bf\Sigma}_0{\bm u}])}{{\bm u}^\sT {\bf\Sigma}_0{\bm u}\E[{\bm u}^\sT {\bf\Sigma}_0{\bm u}]}\right|\\
           =&\frac{p}{\tr(\Sigma_0)}\E\left[\left|\frac{u_iu_j}{{\bm u}^\sT {\bf\Sigma}_0{\bm u}}\right|\left|{\bm u}^\sT {\bf\Sigma}_0{\bm u}-\E[{\bm u}^\sT {\bf\Sigma}_0{\bm u}]\right|\right]\\
           \leq &\frac{1}{c_0}\sqrt{\E\left|\frac{u_iu_j}{{\bm u}^\sT {\bf\Sigma}_0{\bm u}}\right|^2\Var({\bm u}^\sT {\bf\Sigma}_0{\bm u})}\\
    \end{align*}
    Then consider the first term. Since 
    \[
    ({\bm u}^\sT {\bf\Sigma}_0{\bm u} )({\bm u}^\sT{\bf \Sigma}_0^{-1}{\bm u} )\geq 1
    \]
    Under Assumption \ref{V0}, we have 
    \begin{align*}
        \E|\frac{u_iu_j}{{\bm u}^\sT {\bf\Sigma}_0{\bm u}}|^2
        \leq &\E|u_iu_j{\bm u}^\sT{\bf \Sigma}_0^{-1}{\bm u}|^2\\
        =&\E[\sum_{k_1,l_1}\sum_{k_2,l_2}\Sigma_{k_1l_1}^{-1}\Sigma_{k_2l_2}^{-1}u_i^2u_j^2u_{k_1}u_{k_2}u_{l_1}u_{l_2}]\\
        = & \frac{105p^2}{\tr(\Sigma_0)^2}\cdot\frac{\sum_{k_1,l_1}\sum_{k_2,l_2}v_{k_1l_1}v_{k_2l_2}}{p(p+2)(p+4)(p+6)}\\
        \leq &\frac{105}{c_0^2}\cdot\frac{p^2T^2}{p(p+2)(p+4)(p+6)}\\
        =& O(p^{-2})
    \end{align*}
    Then consider the second term. Under Assumption \ref{lambda}, by Lemma \ref{mofs}, we have
    \begin{align*}
        \Var(u^\sT \Sigma_0u)
        =&\E[{\bm u}^\sT {\bf\Sigma}_0{\bm u}{\bm u}^\sT {\bf\Sigma}_0{\bm u}]-\E^2[{\bm u}^\sT {\bf\Sigma}_0{\bm u}]\\
        =& \E[\tr({\bm u}^\sT {\bf\Sigma}_0{\bm u}{\bm u}^\sT {\bf\Sigma}_0{\bm u})]-\frac{\tr({\bf \Sigma}_0)^2}{p^2}\\
        =& \E[\tr({\bm u}{\bm u}^\sT {\bf\Sigma}_0{\bm u}{\bm u}^\sT{\bf \Sigma}_0)]-\frac{\tr({\bf \Sigma}_0)^2}{p^2}\\
         =& \tr(\E[{\bm u}{\bm u}^\sT {\bf\Sigma}_0{\bm u}{\bm u}^\sT]{\bf \Sigma}_0)-\frac{\tr({\bf \Sigma}_0)^2}{p^2}\\
         =&\frac{\tr(\tr({\bf \Sigma}_0){\bf \Sigma}_0+2\diag\{{\bf \Sigma}_0\}{\bf \Sigma}_0+2{\bf \Sigma}_0^2)}{p(p+2)}-\frac{\tr({\bf \Sigma}_0)^2}{p^2}\\
         =&\frac{4\tr({\bf \Sigma}_0^2)}{p(p+2)}+2\frac{\tr({\bf \Sigma}_0)^2}{p^2(p+2)}\\
         \leq& \frac{4\sum_{i=1}^p\lambda_i^2({\bf \Sigma}_0)}{p(p+2)}+\frac{2}{c_0^2(p+2)}\\
        \leq& \frac{4}{\eta^2(p+2)}+\frac{2}{c_0^2(p+2)}\\
        =&O(p^{-1})
    \end{align*}
    A hidden calculation involved here is
    \begin{align*}
          \E[{\bm u}^\sT {\bf\Sigma}_0{\bm u} {\bm u}{\bm u}^\sT]_{ij}
          &=\sum_{k,l}\Sigma_{kl}(\delta_{ij}\delta_{kl}+\delta_{ik}\delta_{jl}+\delta_{il}\delta_{jk})\\
          &=(\tr({\bf \Sigma}_0)+2\Sigma_{ii}+2\Sigma_{jj})\delta_{ij}+(\Sigma_{ij}+\Sigma_{ji})(1-\delta_{ij})
    \end{align*}

    Combine them together, the proof is finished.
\end{proof}

\begin{lemma}\label{Sinfbound}
    Under Assumption \ref{V0}, and \ref{lambda}, 
    \[
    \|{\bf S}\|_\infty=O(p^{-1})
    \]
    Particularly, when $p$ is sufficiently large, there exists a constant $C_{c_0,\eta,M}$, such that
    \[
    \|{\bf S}\|_\infty\leq \frac{C_{c_0,\eta,T,M }}{p}
    \]
\end{lemma}
\begin{proof}
        Consider the representation of $S_{ij}$, when $p$ is sufficiently large, 
    \begin{align*}
            |S_{ij}|=&\left|\sum_{k=1}^p\sum_{l=1}^p{ \Gamma}_{ik}{ \Gamma}_{lj}\E[\frac{u_ku_l}{{\bm u}^\sT {\bf\Sigma}_0{\bm u}}]\right|\\
            \leq& \left|\sum_{k=1}^p\sum_{l=1}^p{ \Gamma}_{ik}{\Gamma}_{lj}\right|O(p^{-1})\\
            \leq&\left|\sqrt{\sum_{k=1}^p({ \Gamma}_{ik})^2}\sqrt{\sum_{l=1}^p({ \Gamma}_{lj})^2}\right|O(p^{-1})\\
            =&\left|\sqrt{{ \Sigma}_{ii}}\sqrt{{ \Sigma}_{jj}}\right|O(p^{-1})\\
            \leq & MO(p^{-1})\\
            \leq & \frac{C_{c_0,\eta,T,M }}{p}
    \end{align*}
\end{proof}
\begin{lemma}{(Approximation Error)}\label{error}
Under Assumption \ref{V0} and \ref{lambda},
    \[
    \|p{\bf S}-{\bf \Lambda}_0\|_\infty = O(p^{-\frac{1}{2}}), \quad p\rightarrow\infty
    \]
    Particularly, when $p$ is sufficiently large, the constant $C_{c_o,\eta,T,M}$ mentioned in Lemma \ref{Sinfbound} can be used to bound the error, i.e.
    \[
    \|p{\bf S}-{\bf \Lambda}_0\|_\infty\leq \frac{C_{c_0,\eta,T,M}}{\sqrt{p}}
    \]
\end{lemma}
\begin{proof}
    Consider the representation of $(pS-\Lambda_0)_{ij}$,
    \begin{align*}
        |(p{\bf S}-{\bf \Lambda}_0)_{ij}|
        =&p\left|{\bf S}-\frac{1}{\tr({\bf \Sigma}_0)}{\bf \Sigma}_0\right|_{ij}\\
        \leq&p\sum_{k=1}^p\sum_{l=1}^p\Gamma_{ik}\Gamma_{lj}\left|\E[\frac{u_ku_l}{{\bm u}^\sT {\bf\Sigma}_0{\bm u}}]-\frac{1}{\tr({\bf \Sigma}_0)}\delta_{kl}\right|\\
        \leq&p\left|\sum_{k=1}^p\sum_{l=1}^p\Gamma_{ik}\Gamma_{lj}\right|O(p^{-\frac{3}{2}})\\
        \leq&\left|\sqrt{\sum_{k=1}^p(\Gamma_{ik})^2}\sqrt{\sum_{l=1}^p(\Gamma_{lj})^2}\right|O(p^{-\frac{1}{2}})\\
        \leq & \left|\sqrt{\Sigma_{ii}}\sqrt{\Sigma_{jj}}\right|O(p^{-\frac{1}{2}})\\
        \leq & MO(p^{-\frac{1}{2}})
    \end{align*}
    Compare the coefficient, and then the proof is finished.
\end{proof}
\begin{lemma}\label{SSestimation}
Under Assumption \ref{A1A2}, there exists $C>0$, such that for  sufficiently large $n$ that satisfies
\[
    n>\max\{\left[\log(p)+\log(\sqrt{\frac{2}{\alpha}})\right]^{\frac{1}{3}}, \left[\log(p)+\log({\frac{2}{\alpha}})\right]^{\frac{1}{3}}\}
\]
we have
\[
\|\widehat{\bf S} - {\bf S}\|_{\infty} \leq C \left( \frac{8 \lambda_1({\bf\Sigma}_0)}{p \lambda_p({\bf\Sigma}_0)} + \|{\bf S}\|_{\infty} \right) 
\sqrt{\frac{\log p + \log(1/\sqrt{\alpha/2})}{n}}.
\] 
\[
    \|\hat{\bm \mu} - {\bm \mu}\|_\infty \leq \frac{2 \lambda_1({\bf \Sigma}_0)}{\zeta_1 p \lambda_p({\bf \Sigma}_0)} \sqrt{\frac{\log (p) + \log(2/\alpha)}{n}}
\]
with probability larger than $1-\alpha$.
\end{lemma}
\begin{proof}
    The following proof is an elaborated version of the proof of Theorem 7.1 in \cite{feng2024spatialsignbasedprincipal}. 
    
    A standard result is that for ${\bm X}\sim{EC}_p({\bm \mu},{\bf \Sigma}_0,r)$,  $\forall {\bm v}\in \mathbb{S}^{p-1}$
    \begin{equation}\label{subgnorm}
            \|{\bm v}U({\bm X}-{\bm \mu})\|_{\psi_2}\leq\sqrt{\frac{\lambda_1({\bf \Sigma}_0)}{\lambda_p({\bf \Sigma}_0)}\cdot\frac{2}{p}}
    \end{equation}

    This can be viewed as a corollary of Theorem 4.2 in \cite{han2016ecahighdimensionalelliptical}, with a fact that $U({\bm X})\deql S({\bm X})$, where $S(\cdot)$ is the self-normalized operator, if ${\bm X}$ is of elliptical distribution with mean $0$.

    By the inequality $\|{\bm X\bm Y}\|_{\psi_1}\leq 2\|{\bm X}\|_{\psi_2}\|{\bm Y}\|_{\psi_2}$, we can know that,
    \[
    \|\be_j^\sT U({\bm X}_i-{\bm \mu}_i)U({\bm X}_i-{\bm \mu}_i)^\sT\be_k\|_{\psi_1}\leq \frac{\lambda_1({\bf \Sigma}_0)}{\lambda_p({\bf \Sigma}_0)}\cdot\frac{4}{p}\leq \frac{\lambda_1({\bf \Sigma}_0)}{\lambda_p({\bf \Sigma}_0)}\cdot\frac{8}{p}
    \]
    By Berstein's inequality(refer to \cite{Vershynin_2018}), when $t<n(\frac{\lambda_1({\bf \Sigma}_0)}{\lambda_p({\bf \Sigma}_0)}\cdot\frac{8}{p}+\|{\bf S}\|_\infty)$, we have
    \[
    \P(\left|\frac{1}{n}\sumn\left[\be_j^\sT U({\bm X}_i-{\bm \mu})U({\bm X}_i-{\bm \mu})^\sT\be_k-S_{jk}\right]\right|\geq t)
    \leq 2 \exp\{-c\frac{nt^2}{(\frac{\lambda_1({\bf\Sigma}_0)}{\lambda_p({\bf \Sigma}_0)}\cdot\frac{8}{p}+S_{jk})^2}\}
    \]
    Let $\tilde{\bf S}=\frac{1}{n}\sumn[U({\bm X}_i-{\bm \mu})U({\bm X}_i-{\bm \mu})^\sT]$, we then have when $t<n(\frac{\lambda_1({\bf \Sigma}_0)}{\lambda_p({\bf \Sigma}_0)}\cdot\frac{8}{p}+\|{\bf S}\|_\infty)$, 
    \[
    \P(\|\tilde{\bf S}-{\bf S}\|_\infty\geq t)\leq 2p^2\exp\{-c\frac{nt^2}{(\frac{\lambda_1({\bf\Sigma}_0)}{\lambda_p({\bf \Sigma}_0)}\cdot\frac{8}{p}+\|{\bf S}\|_{\infty})^2}\}
    \]
    Therefore, for sufficiently large $n$ and with probability greater than $1-\alpha$, we have 
    \[
    \|\tilde{\bf S}-{\bf S}\|_\infty\leq \sqrt{\frac{2}{c}}(\frac{\lambda_1({\bf \Sigma}_0)}{\lambda_p({\bf \Sigma}_0)}\cdot\frac{8}{p}+\|{\bf S}\|_{\infty})\sqrt{\frac{\log(p)+\log(1/\sqrt{\alpha/2})}{n}}
    \]
    Next, we will show the bound of $\| \hat{\bf S} - {\bf S} \|_{\infty}$. For simplicity, we denote $U({\bm X}_i-\hat {\bm \mu})$ by $\hat {\bm U}_i$ and $U({\bm X}_i- {\bm \mu})$ by ${\bm U}_i$. Following the notations in Assumption \ref{A1A2}, we denote $\|{\bm X}_i-\hat{\bm \mu}\|_2$ by $\hat \xi_i$. Because
\begin{align*}
    \hat{\bf S} - \tilde{\bf S} =& \frac{1}{n} \sum_{i=1}^n \left[ \hat{\bm U}_i\hat {\bm U}_i^\sT - {\bm U}_i{\bm U}_i^\sT \right],\\
=& \frac{2}{n} \sum_{i=1}^n (\hat {\bm U}_i -{\bm U}_i) {\bm U}_i^\sT + \frac{1}{n} \sum_{i=1}^n (\hat {\bm U}_i -{\bm U}_i)(\hat {\bm U}_i-{\bm U}_i)^T \\
\triangleq &J_1 + J_2.
\end{align*}

And
\begin{align*}
    \hat{\bm U}_i - {\bm U}_i &= \frac{{\bm X}_i - \hat{\bm \mu}}{\|{\bm X}_i - \hat{\mu}\|_2} - \frac{{\bm X}_i - {\bm \mu}}{\|{\bm X}_i - {\bm \mu}\|_2}
    \\
    &=\frac{{\bm X}_i - \hat{\bm \mu}}{\|{\bm X}_i - \hat{\bm \mu}\|_2} - \frac{{\bm X}_i - \hat{\bm\mu}}{\|{\bm X}_i - {\bm \mu}\|_2} + \frac{{\bm X}_i - \hat{\bm \mu}}{\|{\bm X}_i - {\bm \mu}\|_2}- \frac{{\bm X}_i - {\bm \mu}}{\|{\bm X}_i - {\bm \mu}\|_2}
    \\
    &=({\bm X}_i-\hat{\bm \mu})(\hat\xi_i^{-1}- \xi_i^{-1})+(\hat{\bm \mu}-{\bm \mu})\xi_i^{-1}
    \\
    &={\bm U}_i(\xi_i\hat\xi_i^{-1}-1)+(\hat {\bm \mu}-{\bm \mu})(\hat\xi_i^{-1}- \xi_i^{-1})+(\hat{\bm \mu}-{\bm \mu})\xi_i^{-1}
\end{align*}

Hence, we can rewrite $\|\frac{J_1}{2}\|_{\infty}$ as follows.
\begin{align*}
    \|\frac{J_1}{2}\|_{\infty} &\leq \frac{1}{n} \left\|\sum_{i=1}^n (\xi_i \hat{\xi}_i^{-1} - 1) {\bm U}_i {\bm U}_i^\top\right\|_{\infty} 
+ \frac{1}{n} \left\|\sum_{i=1}^n (\hat{\bm \mu} - {\bm \mu}) (\hat{\xi}_i^{-1}-{\xi}_i^{-1}) {\bm U}_i^\top\right\|_{\infty}+ \frac{1}{n} \left\|\sum_{i=1}^n (\hat{\bm \mu} - {\bm \mu}) {\xi}_i^{-1} {\bm U}_i^\top\right\|_{\infty}\\
&\triangleq J_{11}+J_{12}+J_{13}
\end{align*}

First, 
\[
J_{11} \leq \max_{1 \leq i \leq n} |\xi_i \hat{\xi}_i^{-1} - 1| \|\tilde{\bf S}\|_{\infty}.
\]

Because $|\hat{\xi}_i - \xi_i| \leq \|\hat{\bm\mu} - {\bm \mu}\|_2$, so $|\xi_i \hat{\xi}_i^{-1} - 1| \leq \frac{\|\hat{\bm \mu} - {\bm \mu}\|_2}{\xi_i- \|\hat {\bm \mu} - {\bm \mu}\|_2}$.

Then,
\[
\max_{1 \leq i \leq n} |\xi_i \hat{\xi}_i^{-1} - 1| \leq \zeta_1 \|\hat{\bm \mu} - {\bm \mu}\|_2 \max_{1 \leq i \leq n} (\zeta_1\xi_i - \zeta_1 \|\hat{\bm \mu} - {\bm \mu}\|_2)^{-1}.
\]

By Lemma 7.4 in \cite{feng2024spatialsignbasedprincipal}, we have $\zeta_1 \|\hat{\bm \mu} - {\bm \mu}\|_2 = O_p(n^{-1/2})$. And by the sub-Gaussian assumption of $\nu_i$, we have $\max_{1 \leq i \leq n} \nu_i = O_p(\log^{1/2} n)$. So $\max_{1 \leq i \leq n} |\xi_i \hat{\xi}_i^{-1} - 1| = O_p\left(\sqrt{\frac{\log n}{n}}\right)$.

In addition,
\[
\|\tilde{\bf S}\|_{\infty} \leq \|\tilde{\bf S} - {\bf S}\|_{\infty} + \|{\bf S}\|_{\infty}.
\]

So, for sufficiently large $n$ and $c^*$, with probability larger than $1 - 2\alpha$,
\[
J_{11} \leq c^* \sqrt{\frac{\log n}{n}} \left(\|{\bf S}\|_{\infty} + \sqrt{\frac{2}{c}} \left(\frac{8 \lambda_1({\bf \Sigma}_0)}{p \lambda_p({\bf \Sigma}_0)}+\|{\bf S}\|_{\infty}\right)  \sqrt{\frac{\log (p) + \log(1/\sqrt{\alpha/2})}{n}}\right).
\]
Hence, $J_{11}=o_p(\|\tilde{\bf S}-{\bf S}\|_\infty)$.

Second, 
\[
J_{13} \leq \zeta_1 \|\hat{\bm \mu} - {\bm \mu}\|_\infty \cdot \left\|\frac{1}{n} \sum_{i=1}^T \nu_i {\bm U}_i \right\|_\infty.
\]

By \eqref{subgnorm} and \( \nu_i \) is also a sub-Gaussian variable with parameter \( K_\nu \), we have
\[
\| \nu_i {\bm U}_i^\top {\bm v} \|_{\psi_1} \leq \| \nu_i \|_{\psi_2} \| {\bm U}_i^\top {\bm v} \|_{\psi_2} \leq K_\nu \sqrt{\frac{\lambda_1({\bf\Sigma}_0)}{\lambda_p({\bf \Sigma}_0)}\frac{2}{p}}.
\]

So, by the Bernstein inequality (refer to \cite{Vershynin_2018}), we have for $k\in[p]$ and \( t \leq n K_\nu \sqrt{2\lambda_1({\bf \Sigma}_0)/p \lambda_p({\bf \Sigma}_0)} \),

\[
\P\left(\left|\frac{1}{n} \sum_{i=1}^n \nu_i {\bm U}_i^\top e_k \right|\geq t \right) \leq 2\exp\left(-c\frac{nt^2}{2K_\nu^2 \lambda_1({\bf \Sigma}_0)/p \lambda_p({\bf \Sigma}_0)}\right).
\]

Then,
\[
\P\left(\left\|\frac{1}{n} \sum_{i=1}^n \nu_i {\bm U}_i\right\|_\infty \geq t\right) \leq 2p\exp\left(-\frac{n t^2}{2 K_\nu^2 \lambda_1({\bf \Sigma}_0)/p \lambda_p({\bf \Sigma}_0)}\right) \leq \alpha.
\]
by setting 
\[
t = \sqrt{\frac{2 K_\nu^2 \lambda_1({\bf\Sigma}_0)}{p \lambda_p({\bf \Sigma}_0)} \frac{\log (p) + \log(2/\alpha)}{n}}.
\]

Similar to the above arguments, we have
\[
\P\left(\left\|\frac{1}{n} \sum_{i=1}^n {\bm U}_i\right\|_\infty \geq \sqrt{\frac{2 \lambda_1({\bf \Sigma}_0)}{p \lambda_p({\bf \Sigma}_0)} \frac{\log (p) + \log(2/\alpha)}{n}} \right) \leq \alpha.
\]

By the equation \( \sum_{i=1}^n \hat{\bm U}_i = 0 \), we have
\[
\zeta_1 (\hat{\bm \mu} -{\bm \mu}) = \frac{1}{n} \sum_{i=1}^n {\bm U}_i + \frac{1}{n} \sum_{i=1}^n (\xi_i \hat{\xi}_i^{-1} - 1) {\bm U}_i +\left(\frac{\zeta_1 }{\frac{1}{n} \sum_{i=1}^n \hat\xi_i^{-1}}-1 \right) \frac{1}{n} \sum_{i=1}^n {\bm U}_i
\]
\[
+ \left(\frac{\zeta_1 }{\frac{1}{n} \sum_{i=1}^n \hat\xi_i^{-1}}-1 \right) \frac{1}{n} \sum_{i=1}^n {\bm U}_i (\xi_i\hat\xi_i^{-1} - 1) \overset{\triangle}{=} B_1 + B_2 + B_3 + B_4.
\]

By the sub-Gaussian assumption of \( \nu_i \), we have \( \| B_i \|_\infty = O_p\left(\sqrt{\frac{\log n}{n}}\right)\|B_1\|_\infty, i = 2, 3, 4 \) . So
\[
\zeta_1\|\hat{\bm \mu} - {\bm \mu}\|_\infty \leq \frac{2 \lambda_1({\bf \Sigma}_0)}{p \lambda_p({\bf \Sigma}_0)} \sqrt{\frac{\log (p) + \log(2/\alpha)}{n}}
\]
with probability greater than $1-\alpha$, for sufficiently large $n$.

Therefore, there exists $C^*$, s.t,
\[
J_{13}\leq C^*\frac{\lambda_1({\bf \Sigma}_0)}{p\lambda_p({\bf \Sigma}_0)}\frac{\log(p)+\log(2/\alpha)}{n}
\]
which implies $J_{13}=o_p(\|\tilde{\bf S}-{\bf S}\|_\infty)$.

Finally, 
\[
J_{12}\leq \max_{1\leq i\leq n}|\xi_i\hat\xi_i^{-1}-1|J_{13}
\]
which implies $J_{12}$ is of smaller order than $J_{13}$. By the triangle inequality, we can obtain $\|J_1\|_{\infty}=o_p(\|\tilde{\bf S}-{\bf S}\|_\infty)$. Similarly we can know that $\|J_2\|_{\infty}=o_p(\|\tilde{\bf S}-{\bf S}\|_\infty)$. 

Above all, we can know that with probability greater than $1-\alpha$, there exists $C>0$, such that
\[
\|\hat {\bf S}-{\bf S}\|_\infty \leq C(\frac{8\lambda_1({\bf \Sigma}_0)}{p\lambda_p({\bf \Sigma}_0)}+\|{\bf S}\|_\infty)\sqrt{\frac{\log(p)+\log(1/\sqrt{\alpha/2})}{n}}
\]
\end{proof}

\begin{lemma}\label{infdiffestimation}
     Suppose $\|p\hat {\bf S}{\bf V}_0-\id_p\|_\infty\leq \lambda_n$, and then $\hat {\bf V}_{SCLIME}$ defined in \eqref{CSSIME} satisfies the following inequality.
    \[
    \|\hat {\bf V}_{SCLIME}-{\bf V}_0\|_\infty\leq3\|{\bf V}_0\|_{L_1}^2\|{\bf \Lambda}_0-p\hat {\bf S}\|_\infty+\|{\bf V}_0\|_{L_1}\lambda_n
    \]
\end{lemma}
\begin{proof}
    First, $\|p\hat {\bf S} {\bf V}_0-\id_p\|_\infty\leq \lambda_n$ and  Lemma 1 of \cite{cai2011constrainedl1minimizationapproach} 
    implies $ \|\hat {\bf V}_1\|_{L_1}\leq \|{\bf V}_0\|_{L_1}$. By definition, we have $\|\hat {\bf V}\|_{L_1}\leq \|\hat {\bf V}_1\|_{L_1}$, Hence, we have
    \[
    \|\hat {\bf V}_{SCLIME}\|_{L_1}\leq \|\hat {\bf V}_1\|_{L_1}\leq \|{\bf V}_0\|_{L_1}
    \]
    
    Therefore, by Lemma \ref{L1ineq}, 
    \begin{align*}
        \|\hat {\bf V}_{SCLIME}-{\bf V}_0\|_\infty 
        \leq\|\hat {\bf V}_1-{\bf V}_0\|_\infty
        &\leq\|{\bf V}_0\|_{L_1}\|{\bf \Lambda}_0(\hat {\bf V}_1-{\bf V}_0)\|_\infty\\
        &\leq \|{\bf V}_0\|_{L_1}\left(\|({\bf \Lambda}_0-p\hat {\bf S})(\hat {\bf V}_1-{\bf V}_0)\|_\infty+\|p\hat {\bf S}(\hat {\bf V}_1-{\bf V}_0)\|_\infty\right)\\
        &\leq \|{\bf V}_0\|_{L_1}\left(\|({\bf \Lambda}_0-p\hat {\bf S})\|_\infty\|(\hat {\bf V}_1-{\bf V}_0)\|_{L_1}+\|p\hat {\bf S}\hat {\bf V}_1-\id_p\|_\infty+ \|\id_p-p\hat {\bf S}{\bf V}_0\|_\infty\right)\\
        &\leq 2\|{\bf V}_0\|_{L_1}^2\|({\bf \Lambda}_0-p\hat {\bf S})\|_\infty+\|{\bf V}_0\|_{L_1}\lambda_n+\|{\bf V}_0\|_{L_1}\|{\bf \Lambda}_0-p\hat {\bf S}\|_\infty\|{\bf V}_0\|_{L_1}\\
        &=3\|{\bf V}_0\|_{L_1}^2\|{\bf \Lambda}_0-p\hat {\bf S}\|_\infty+\|{\bf V}_0\|_{L_1}\lambda_n
    \end{align*}
\end{proof}
\subsubsection{Proof of Theorem \ref{main}}\label{appendix: proof main}
\begin{proof}
For \eqref{inftyerror}, by Lemma \ref{SSestimation}, we know that when $n$ is sufficiently large that satisfy $n>(3\log p)^{\frac{1}{3}}$, with probability larger than $1-2p^{-2}$,
\begin{align*}
\|\widehat{\bf S} - {\bf S}\|_{\infty} 
&\leq C \left( \frac{8 \lambda_1({\bf \Sigma}_0)}{p \lambda_p({\bf \Sigma}_0)} + \|{\bf S}\|_{\infty} \right) 
\sqrt{\frac{\log p + \log(1/\sqrt{2p^{-2}/2})}{n}}\\
&\leq \sqrt{2} C \frac{8 +\eta^2C_{c_0,\eta,T,M } }{p\eta^2}
\sqrt{\frac{\log p }{n}}
\end{align*}

Next, consider when $p$ is sufficiently large, there exist $C_1$ and $C_\epsilon$, such that
\begin{align*}
    \|{\bf \Lambda}_0-p\hat {\bf S}\|_\infty
    \leq& \|{\bf \Lambda}_0-p{\bf S}\|_\infty+p\|{\bf S}-\hat {\bf S}\|_\infty\\
    \leq&  \frac{\sqrt{2}(8 +\eta^2C_{c_0,\eta,T,M })C }{\eta^2}
\sqrt{\frac{\log p }{n}}+\frac{C_{c_0,\eta,T,M }}{\sqrt{p}}.
\end{align*}
Since $\|p\hat {\bf S}{\bf V}_0-\id_p\|_\infty\leq \|p\hat {\bf S} -{\bf \Lambda}_0\|_\infty \|{\bf V}_0\|_{L_1}\leq \lambda_n$, 
by Lemma \ref{infdiffestimation}, the conclusion holds.

For \eqref{L1error}, let $t_n = \|\hat {\bf V}_{SCLIME} - {\bf V}_0\|_\infty$ and define
\[
\bm{h}_j = \hat{\bm v}_j - {\bm v}_j^0,
\]
\[
\bm{h}_j^1 = (\hat{v}_{ij} \1\{|\hat{v}_{ij}| \geq 2t_n\}; 1 \leq i \leq p)^\sT - {\bm v}_j^0, \quad
\]
\[
\bm{h}_j^2 = \bm{h}_j - \bm{h}_j^1.
\]
By the definition \eqref{CSSIME} of $\hat{\bf V}_{SCLIME}$, we have $\|\hat{\bm v}_j\|_1 \leq \|{\bm v}_j^0\|_1$. Then
\[
\|{\bm v}_j^0\|_1 - \|\bm{h}_j^1\|_1 + \|\bm{h}_j^2\|_1 \leq \|{\bm v}_j^0 + \bm{h}_j^1\|_1 + \|\bm{h}_j^2\|_1 = \|\hat{\bm v}_j\|_1 \leq \|{\bm v}_j^0\|_1,
\]
which implies that $|\bm{h}_j^2|_1 \leq |\bm{h}_j^1|_1$. It follows that $|\bm{h}_j|_1 \leq 2|\bm{h}_j^1|_1$. Thus we only need to upper bound $|\bm{h}_j^1|_1$. By triangle inequality and Assumption \ref{V0}, we have
\begin{align*}
\|\bm{h}_j^1\|_1 
&= \sum_{i=1}^p \left|\hat{v}_{ij} \1\{|\hat{v}_{ij}| \geq 2t_n\} - v_{ij}^0\right|\\
&\leq \sum_{i=1}^p \left|v_{ij}^0 \1\{|v_{ij}^0| \leq 2t_n\}\right|
+ \sum_{i=1}^p \left|\hat{v}_{ij} \1\{|\hat{v}_{ij}| \geq 2t_n\} - v_{ij}^0 \1\{|v_{ij}^0| \geq 2t_n\}\right|\\
&\leq (2t_n)^{1 - q} s_0(p) + t_n \sum_{i=1}^p \1\{|\hat{v}_{ij}| \geq 2t_n\}
+ \sum_{i=1}^p |v_{ij}^0| |\1\{|v_{ij}^0| \geq 2t_n\} - \1\{|v_{ij}^0| \geq 2t_n\}|\\
&\leq (2t_n)^{1 - q} s_0(p) + t_n \sum_{i=1}^p \1\{|\hat{v}_{ij}| \geq t_n\}
+ \sum_{i=1}^p |v_{ij}^0| \1\{||v_{ij}^0| - 2t_n| \leq |\hat{v}_{ij} - v_{ij}^0|\}\\
&\leq (2t_n)^{1-q} s_0(p) + (t_n)^{1-q} s_0(p) + (3t_n)^{1-q} s_0(p)\\
&\leq (1 + 2^{1-q} + 3^{1-q}) t_n^{1-q} s_0(p)
\end{align*}

where we use the following inequality: for any $a, b, c \in \mathbb{R}$, we have
\[
\left|\1\{a < c\} - \1\{b < c\}\right| \leq \1\{|b - c| < |a - b|\}.
\]

This completes the proof of \eqref{L1error}.

Finally, \eqref{Ferror} follows from \eqref{inftyerror}, \eqref{L1error}, and the inequality
\[
\|\mathbf{A}\|_F^2 \leq p \|\mathbf{A}\|_{L_1} \|\mathbf{A}\|_\infty
\]
for any $p \times p$ matrix.

\end{proof}
\subsubsection{Proof of Theorem \ref{Graphthm}}\label{appendix: proof graph}
\begin{proof}
    By Theorem \ref{main}, we know that when $n, p$ is sufficiently large, with probability larger than $1-p^{-2}$, 
    \[
    \|\hat {\bf V}_{SCLIME}-{\bf V}_0\|_\infty < \tau_n
    \]
    For $(i,j)\notin S({\bf V}_0)$, we know that 
    \[|\tilde{v}_{ij}|\leq |\tilde{v}_{ij}-v^0_{ij}|+|v_{ij}^0|\leq \max_{i,j}|\hat{v}_{ij}-v_{ij}^0|< \tau_n
    \]
    which implies $\tilde{v}_{ij}=0$.

    For $(i,j)\in S({\bf V}_0)$, since $\theta_{ij}>2\tau_n$, we know that $|\hat v_{ij}|\geq |v_{ij}^0|-|\hat v_{ij}-v^0_{ij}|>\theta_{ij}-\tau_n>\tau_n$, which implies $\tilde{v}_{ij}=\hat v_{ij}$. When $v_{ij}^0>0$, we have $v_{ij}^0\geq \theta_n$, and thus $\tilde{v}_{ij}=\hat v_{ij}\geq v_{ij}^0-|\hat v_{ij}-v^0_{ij}|>\theta_{ij}-\tau_n>\tau_n$, which implies $\text{sgn}(v_{ij}^0)=\text{sgn}(\tilde{v}_{ij})$. When $v_{ij}^0<0$, we have $v_{ij}^0\leq -\theta_n$, and thus $\tilde{v}_{ij}=\hat v_{ij}\leq v_{ij}^0+|\hat v_{ij}-v^0_{ij}|<-\theta_{ij}+\tau_n<-\tau_n$, which implies $\text{sgn}(v_{ij}^0)=\text{sgn}(\tilde{v}_{ij})$.

    Above all, the proof is finished.

\end{proof}
\subsubsection{Proof of Theorem \ref{LDATheorem}}\label{appendix: proof LDA}
\begin{proof}
    First, by Lemma \ref{SSestimation}, for sufficiently large $n$, 
    \begin{align*}
        \|\hat {\bm \mu}_d-{\bm \mu}_d\|_1
        \leq& \frac{1}{2}(\|\hat{\bm \mu}_1-{\bm \mu}_1\|_1+\|\hat{\bm \mu}_0-{\bm \mu}_0\|_1)\\
        \leq& \frac{p}{2}(\|\hat{\bm \mu}_1-{\bm \mu}_1\|_\infty+\|\hat{\bm \mu}_0-{\bm \mu}_0\|_\infty)\\
        \leq &\frac{1}{2}[\frac{2 \lambda_1({\bf \Sigma}_0)}{\zeta_1  \lambda_p({\bf \Sigma}_0)} \sqrt{\frac{\log (p) + \log(2/\alpha)}{n_0}}+\frac{2 \lambda_1({\bf \Sigma}_0)}{\zeta_1  \lambda_p({\bf \Sigma}_0)} \sqrt{\frac{\log (p) + \log(2/\alpha)}{n_1}}]\\
        \approx& \frac{2\sqrt{2} \lambda_1({\bf \Sigma}_0)}{\zeta_1 \lambda_p({\bf \Sigma}_0)} \sqrt{\frac{\log (p) + \log(2/\alpha)}{n}}\\
        \leq &\frac{2\sqrt{2} }{\zeta_1 \eta^2} \sqrt{\frac{\log (p) + \log(2/\alpha)}{n}}
    \end{align*}
    with probability greater than $1-\alpha$.

    Second, by Lemma \ref{SSestimation}, similar as the argument in the proof of Theorem \ref{main}, for sufficiently large $n, p$, with probability greater than $1-\alpha$, if $\lambda_n\geq T \left[\frac{\sqrt{2}(8 +\eta^2C_{c_0,\eta,T,M })C }{\eta^2}
\sqrt{\frac{\log p + \log(1/\sqrt{\alpha/2})}{n}}+\frac{C_{c_0,\eta,T,M }}{\sqrt{p}}\right]$
\[
        \|\hat {\bf V}_{SCLIME}-{\bf V}_0\|_\infty 
        \leq   4T^2\left[\frac{\sqrt{2}(8 +\eta^2C_{c_0,\eta,T,M })C }{\eta^2}
\sqrt{\frac{\log p + \log(1/\sqrt{\alpha/2})}{n}}+\frac{C_{c_0,\eta,T,M }}{\sqrt{p}}\right]
\]

Hence, by choosing $\alpha=2p^{-2}$ and $\lambda_n=T\left[\frac{2(8 +\eta^2C_{c_0,\eta,T,M })C }{\eta^2}
\sqrt{\frac{\log p }{n}}+\frac{C_{c_0,\eta,T,M }}{\sqrt{p}}\right]$, with probability greater than $1-4p^{-2}$, for sufficiently large $n, p$ that satisfy $n>(3\log p)^\frac{1}{3}$, 
\[
            \|\hat {\bf V}_{SCLIME}-{\bf V}_0\|_\infty 
        \leq   4T^2\left[\frac{2(8 +\eta^2C_{c_0,\eta,T,M })C }{\eta^2}
\sqrt{\frac{\log p }{n}}+\frac{C_{c_0,\eta,T,M }}{\sqrt{p}}\right]
\]
\[
    \|\hat {\bm \mu}_d-{\bm \mu}_d\|_1\leq  \frac{2\sqrt{6} }{\zeta_1 \eta^2} \sqrt{\frac{\log p }{n}}
\]
Therefore, 
   \begin{align*}
        \|\hat {\bf V}_{SCLIME}\hat {\bm \mu}_d-{\bf V}_0{\bm \mu}_d\|_\infty
        =&\|\hat {\bf V}_{SCLIME}\hat {\bm \mu}_d-\hat {\bf V}_{SCLIME}{\bm \mu}_d+\hat {\bf V}_{SCLIME}{\bm \mu}_d-{\bf V}_0{\bm \mu}_d\|_\infty\\
        \leq & \|\hat {\bf V}_{SCLIME}\|_\infty\|\hat{\bm \mu}_d-{\bm \mu}_d\|_1+\|\hat {\bf V}_{SCLIME}-{\bf V}_0\|_\infty\|{\bm \mu}_d\|_1\\
        \leq & \| {\bf V}_0\|_\infty\|\hat{\bm {\bm \mu}}_d-{\bm \mu}_d\|_1 + \|\hat {\bf V}_{SCLIME}-{\bf V}_0\|_\infty(\|\hat{\bm \mu}_d-{\bm \mu}_d\|_1+\|{\bm \mu}_d\|_1)\\
        \leq & T\frac{2\sqrt{6} }{\zeta_1 \eta^2} \sqrt{\frac{\log p }{n}}+ 4M^*T^2[\frac{2(8 +\eta^2C_{c_0,\eta,T,M })C }{\eta^2}
\sqrt{\frac{\log p }{n}}+\frac{C_{c_0,\eta,T,M }}{\sqrt{p}}]\\
+&o(\|\hat {\bf V}_{SCLIME}-{\bf V}_0\|_\infty )
    \end{align*}
Above all , the final result is if $\lambda_n = T[\frac{2(8 +\eta^2C_{c_0,\eta,T,M })C }{\eta^2}
\sqrt{\frac{\log p }{n}}+\frac{C_{c_0,\eta,T,M }}{\sqrt{p}}]$, for sufficiently large $n, p$, there exists $C_*$, such that 
\begin{align*}
    \|\hat {\bf V}_{SCLIME}\hat {\bm \mu}_d-{\bf V}_0{\bm \mu}_d\|_\infty
     &\leq C_*\left[T\frac{2\sqrt{6} }{\zeta_1 \eta^2} \sqrt{\frac{\log p }{n}}+ 4M^*T^2\left[\frac{2(8 +\eta^2C_{c_0,\eta,T,M })C }{\eta^2}
\sqrt{\frac{\log p }{n}}+\frac{C_{c_0,\eta,T,M }}{\sqrt{p}}\right]\right]\\
& = \left(\frac{2\sqrt{6} TC_*}{\zeta_1 \eta^2} + \frac{8M^*T^2C_*C(8 +\eta^2C_{c_0,\eta,T,M }) }{\eta^2}\right)
\sqrt{\frac{\log p }{n}}+\frac{4M^*T^2C_*C_{c_0,\eta,T,M }}{\sqrt{p}}
\end{align*}

with probability greater than $1-4p^{-2}$
\end{proof}
\subsection{Appendix B: SGLASSO-related Proofs}
From a high-level view, we want to first show that $\tilde {\bf V}_{SGLASSO}-{\bf V}_0$ can achieve the desirable rate, where
    \[
        \tilde{\bf V}_{SGLASSO}=\argmin_{{\bf V}\succ 0, {\bf V}_{{\bf S}^c}=0}\tr(p\hat {\bf S}{\bf V})-\log\det({\bf V})+\lambda_n\|{\bf V}\|_1
    \]
    We then show $\hat {\bf V}_{SGLASSO}=\tilde{\bf V}_{SGLASSO}$. Before introducing the main proof, we first define some useful concepts and state some useful lemmas. We define ${\bf \Delta}=\tilde{\bf V}_{SGLASSO}-{\bf V}_0$, ${\bf W}=p\hat {\bf S}-{\bf \Lambda}_0$, $R({\bf \Delta})=({\bf V}_0+{\bf \Delta})^{-1}-{\bf V}_0^{-1}+{\bf V}_0^{-1}{\bf \Delta} {\bf V}_0^{-1}$.
\begin{lemma}
    For ${\bf A},{\bf B}\in \R^{p\times p}$, we have
    \[
        \|{\bf A}{\bf B}\|_{L_\infty}\leq \|{\bf A}\|_{L_\infty}\|{\bf B}\|_{L_\infty}
    \]
\end{lemma}
\begin{lemma}
    For ${\bf A}\in \R^{p\times p}$ and symmetric ${\bf B}\in \R^{p\times p}$, we have 
    \[
        {\rm Vec}({\bf B}{\bf A}{\bf B})=({\bf B}\otimes {\bf B}) {\rm Vec}({\bf A})
    \]
\end{lemma}
\begin{lemma}
    $R({\bf \Delta})$ can be represented as ${\bf V}_0^{-1}{\bf \Delta} {\bf V}_0^{-1}{\bf \Delta} J {\bf V}_0^{-1}$, where $J=\sum_{k=0}^\infty(-1)^{k}({\bf V}_0^{-1}{\bf \Delta})^k$. Moreover, its element-wise $l_\infty$ bound can be controlled when $K^*d\|{\bf \Delta}\|_\infty< \frac{1}{3}$.
    \[
        \|R({\bf \Delta})\|_\infty \leq \frac{3}{2}(K^*)^3 d\|{\bf \Delta}\|_\infty^2
    \]
\end{lemma}
\begin{proof}
    Consider
    \begin{align*}
        R({\bf \Delta})&=({\bf V}_0+{\bf \Delta})^{-1}-{\bf V}_0^{-1}+{\bf V}_0^{-1}{\bf \Delta} {\bf V}_0^{-1}\\
        &=(I+{\bf V}_0^{-1}{\bf \Delta})^{-1}{\bf V}_0^{-1}-{\bf V}_0^{-1}+{\bf V}_0^{-1}{\bf \Delta} {\bf V}_0^{-1}\\
        &=\sum_{k=0}^\infty(-1)^{k}({\bf V}_0^{-1}{\bf \Delta})^k{\bf V}_0^{-1}-{\bf V}_0^{-1}+{\bf V}_0^{-1}{\bf \Delta} {\bf V}_0^{-1}\\
        &=\sum_{k=2}^\infty(-1)^{k}({\bf V}_0^{-1}{\bf \Delta})^k{\bf V}_0^{-1}\\
        &={\bf V}_0^{-1}{\bf \Delta} {\bf V}_0^{-1}{\bf \Delta} \sum_{k=0}^\infty(-1)^{k}({\bf V}_0^{-1}{\bf \Delta})^k{\bf V}_0^{-1}\\
        &={\bf V}_0^{-1}{\bf \Delta} {\bf V}_0^{-1}{\bf \Delta} J{\bf V}_0^{-1}
    \end{align*}
    Moreover, consider 
    \begin{align*}
        \|R({\bf \Delta})\|_\infty 
        &= \max_{i,j}|e_i{\bf V}_0^{-1}{\bf \Delta} {\bf V}_0^{-1}{\bf \Delta} J {\bf V}_0^{-1}e_j|\\
        &\leq \max_i\|e_i{\bf V}_0^{-1}{\bf \Delta}\|_\infty\max_{j}\|{\bf V}_0^{-1}{\bf \Delta} J{\bf V}_0^{-1}e_j\|_1\\
        &\leq \max_i \|e_i{\bf V}_0^{-1}\|_1\|{\bf \Delta}\|_\infty \|{\bf V}_0^{-1}J^\sT{\bf \Delta} {\bf V}_0^{-1}\|_{L_\infty}\\
        &\leq \|{\bf V}_0^{-1}\|_{L_\infty}\|{\bf \Delta}\|_\infty \|{\bf V}_0^{-1}J^\sT{\bf \Delta} {\bf V}_0^{-1}\|_{L_\infty}\\
        &\leq (K^*)^3 d\|{\bf \Delta}\|^2_\infty\|J^\sT\|_{L_\infty}
    \end{align*}
    Hence we only have to bound $\|J^\sT\|_{L_\infty}$. Consider
 \begin{align*}
     \|J^\sT\|_{L_\infty}
     &\leq \sum_{k=0}^\infty \|{\bf \Delta} {\bf V}_0^{-1}\|^k_{L_\infty}\\
     & \leq \frac{1}{1-\|{\bf \Delta} \|_{L_\infty}\|{\bf V}_0^{-1}\|_{L_\infty}}\\
     &\leq \frac{1}{1-K^*d\|{\bf \Delta}\|_\infty}\\
     &\leq \frac{3}{2}
 \end{align*}
 Combine them together, the proof is finished.
\end{proof}
\subsubsection{Proof of Theorem \ref{maing}}\label{appendix: proof maing}
\begin{proof}
    The proof can be separated into two parts.  

    \textbf{Step I: Rate of $\tilde{\bf V}_{SGLASSO}-{\bf V}_0$}

    Consider \[
    \nabla^2_{\bf V}(\tr(p\hat {\bf S} {\bf V})-\log\det ({\bf V})) = {\bf V}^{-1}\otimes {\bf V}^{-1}\succ 0
    \]
    then we know 
    \[
        \nabla^2_{{\bf V}_S}(\tr(p\hat {\bf S} {\bf V})-\log\det ({\bf V})) = [{\bf V}^{-1}\otimes {\bf V}^{-1}]_{SS}\succ 0
    \]
    By Lagrangian duality, we know $\tilde{\bf V}_{SGLASSO}=\argmin _{{\bf V}\succ 0, {\bf V}_{S^c}=0, \|{\bf V}\|_1\leq C(\lambda_n)}[\tr(p\hat {\bf S} {\bf V})-\log\det({\bf V})]$. Hence the strict convexity of the objective function in ${\bf V}_S$ implies the uniqueness of $\tilde{\bf V}_{SGLASSO}$. Besides, we know that 
    \[
        {\bf V}_S = [\tilde{\bf V}_{SGLASSO}]_{S} \Longleftrightarrow G({\bf V}_S)\triangleq[p\hat {\bf S}]_S-{\bf V}^{-1}_S+\lambda_n {\bf Z}_S=0
    \]
    where ${\bf Z}_S\in \partial\|{\bf V}_S\|_1$.

    Define another map $F:\R^{|S|}\rightarrow\R^{|S|}$, such that $F(\Vec({\bf \Delta}^\prime_S))\triangleq-({\bf \Gamma}^*_{SS})^{-1}{\rm Vec}(G([V_0]_S+{\bf \Delta}^\prime_S))+{\rm Vec}({\bf \Delta}^\prime_S)$. By definition, we know that $F(\Vec({\bf \Delta}^\prime_S))=\Vec({\bf \Delta}^\prime_S)$ if and only if $\Vec(G([{\bf V}_0]_S+{\bf \Delta}^\prime_S))=0$, i.e. ${\bf \Delta}^\prime_S={\bf \Delta}_S$.

    Denote $2K^*(\|{\bf W}\|_\infty+\lambda_n)$ by $r$ and $\{\Vec({\bf \Delta}^\prime_S): \|{\bf \Delta}^\prime_S\|_\infty \leq r\}$ by $B(r)$, we now \textbf{claim} that $F(B(r))\subseteq B(r)$. Since $B(r)$ is a nonempty compact convex set, then by Brouwer fixed-point theorem, we can know that $\Vec({\bf \Delta}_S) \in B(r)$, and hence
    \[
        \|{\bf \Delta}\|_\infty =\|{\bf \Delta}_S\|_\infty \leq  r
    \]
    Using similar argument as in the proof of Theorem \ref{main}, we could show that when $p$ is sufficiently large, with probability larger than $1-2p^{-2}$
    \[
        \|{\bf W}\|_\infty = \|p\hat {\bf S}-{\bf \Lambda}_0\|_\infty\leq \left(\frac{\sqrt{2}(8 +\eta^2C_{c_0,\eta,s,M })C }{\eta^2}
\sqrt{\frac{\log p }{n}}+\frac{C_{c_0,\eta,s,M }}{\sqrt{p}}\right)
    \]
    where we replace $T$ by $s$ in the absolute constant since in terms of $\|{\bf V}_0\|_1$, $pT $ and $\frac{1}{c_0\eta}(p+s)$ is of the same order. Then we can show that by  picking $\lambda_n =\frac{8}{\alpha}\left(\frac{\sqrt{2}(8 +\eta^2C_{c_0,\eta,s,M })C }{\eta^2}
\sqrt{\frac{\log p }{n}}+\frac{C_{c_0,\eta,s,M }}{\sqrt{p}}\right)$.
    \[
       \| \tilde{\bf V}_{SGLASSO}-{\bf V}_0 \|_\infty \leq 2K^*(\frac{8}{\alpha}+1) \left(\frac{\sqrt{2}(8 +\eta^2C_{c_0,\eta,s,M })C }{\eta^2}
\sqrt{\frac{\log p }{n}}+\frac{C_{c_0,\eta,s,M }}{\sqrt{p}}\right)
    \]

    Now, let us turn to the proof of the \textbf{claim}.

    For ${\bf \Delta}^\prime_S\in B(r)$, since ${\bf \Delta}_{S^c}=0$, we have
    \begin{align*}
        F(\Vec({\bf \Delta}^\prime_S))
        &=-({\bf \Gamma}^*_{SS})^{-1}{\rm Vec}(G([{\bf V}_0]_S+{\bf \Delta}^\prime_S))+{\rm Vec}({\bf \Delta}^\prime_S)\\
        &=-({\bf \Gamma}^*_{SS})^{-1}{\rm Vec}(-(({\bf V}_0+{\bf \Delta}^\prime)^{-1}-{\bf V}_0^{-1})_S+{\bf W}_S+\lambda_n {\bf Z}_S)+{\rm Vec}({\bf \Delta}^\prime_S)\\
        &=-({\bf \Gamma}^*_{SS})^{-1}[{\rm Vec}(-(({\bf V}_0+{\bf \Delta}^\prime)^{-1}-{\bf V}_0^{-1})_S)+({\bf \Gamma}^*_{SS})\Vec({\bf \Delta}^\prime_S)]-({\bf \Gamma}^*_{SS})^{-1}\Vec({\bf W}_S+\lambda_n {\bf Z}_S)\\
        &=-({\bf \Gamma}^*_{SS})^{-1}{\rm Vec}([-(({\bf V}_0+{\bf \Delta}^\prime)^{-1}-{\bf V}_0^{-1})+{\bf V}_0^{-1}{\bf \Delta}^\prime {\bf V}_0^{-1}]_S)-({\bf \Gamma}^*_{SS})^{-1}\Vec({\bf W}_S+\lambda_n {\bf Z}_S)\\
        &= ({\bf \Gamma}^*_{SS})^{-1} \Vec([R({\bf \Delta}^\prime)]_S)-({\bf \Gamma}^*_{SS})^{-1}\Vec({\bf W}_S+\lambda_n {\bf Z}_S)
    \end{align*}
    Hence, since $r\leq \min\{\frac{1}{3(K^*)^4d}, \frac{1}{3K^*d}\}$
    \begin{align*}
    \|F(\Vec({\bf \Delta}^\prime_S))\|_\infty
    &\leq K^*\|R({\bf \Delta}^\prime)\|_\infty+K^*(\|{\bf W}\|_\infty +\lambda_n)\\
    &\leq \frac{3}{2}(K^*)^4d\|{\bf \Delta}^\prime\|_\infty^2+\frac{r}{2}\\
    &= \frac{3}{2}(K^*)^4dr^2+\frac{r}{2}\\
    &\leq r
    \end{align*}
Above all, we show the rate of $\tilde{\bf V}_{SGLASSO}-{\bf V}_0$ in terms of matrix maximum norm.

\textbf{Step II: $\hat{\bf V}_{SGLASSO}=\tilde{\bf V}_{SGLASSO}$}

Similar as the above argument, we know the uniqueness of $\hat {\bf V}_{SGLASSO}$ and 
\[
    {\bf V} = \hat{\bf V}_{SGLASSO} \Longleftrightarrow p\hat {\bf S}-{\bf V}^{-1}+\lambda_n {\bf Z}=0
\]
where ${\bf Z}\in \partial\|{\bf V}\|_1$.

To show $\hat{\bf V}_{SGLASSO}=\tilde{\bf V}_{SGLASSO}$, we only have to show that ${\bf Z}_{S^c}\triangleq\frac{1}{\lambda_n}(-[p\hat {\bf S}]_{S^c}+[\tilde{\bf V}^{-1}_{SGLASSO}]_{S^c})$ satisfies $\|{\bf Z}_{S^c}\|_\infty<1$.

Combine $G([\tilde{\bf V}^{-1}_{SGLASSO}]_{S})=0$ and the above definition together, we have
\[
    p\hat{\bf S}-\tilde{\bf V}_{SGLASSO}^{-1}+\lambda_n {\bf Z}=0
\]
We can rewrite it as 
\[
    {\bf V}_0^{-1}{\bf \Delta} {\bf V}_0^{-1}+{\bf W}-R({\bf \Delta})+\lambda_n {\bf Z}=0
\]
Vectorize the equation, we have
\[
\Vec({\bf V}_0^{-1}{\bf \Delta} {\bf V}_0^{-1}+{\bf W}-R({\bf \Delta})+\lambda_n {\bf Z})= {\bf \Gamma}^* \Vec({\bf \Delta})+\Vec({\bf W})-\Vec(R({\bf \Delta}))+\lambda_n\Vec({\bf Z})=0
\]
Since ${\bf \Delta}_{S^c}=0$, we have
\[
    {\bf \Gamma}^*_{SS} \Vec({\bf \Delta})_S+\Vec({\bf W})_S-\Vec(R({\bf \Delta}))_S+\lambda_n\Vec({\bf Z})_S=0
\]
\[
    {\bf \Gamma}^*_{S^cS} \Vec({\bf \Delta})_{S}+\Vec({\bf W})_{S^c}-\Vec(R({\bf \Delta}))_{S^c}+\lambda_n\Vec({\bf Z})_{S^c}=0
\]
From the above, we have 
\[
    \Vec({\bf \Delta})_S=({\bf \Gamma}^*_{SS})^{-1}[-\Vec({\bf W})_S+\Vec(R({\bf \Delta}))_S-\lambda_n\Vec({\bf Z})_S]
\]
Plug in to another equation, we have 
\begin{align*}
        \Vec({\bf Z})_{S^c}
        &=-\frac{1}{\lambda_n}{\bf \Gamma}^*_{S^cS} \Vec({\bf \Delta})_S+\frac{1}{\lambda_n}\Vec(R({\bf \Delta}))_{S^c}-\frac{1}{\lambda_n}\Vec({\bf W})_{S^c}\\
        &=\frac{1}{\lambda_n}{\bf \Gamma}^*_{S^cS}({\bf \Gamma}^*_{SS})^{-1}[\Vec({\bf W})_S-\Vec(R({\bf \Delta}))_S]+ {\bf \Gamma}^*_{S^cS}({\bf \Gamma}^*_{SS})^{-1}\Vec({\bf Z})_S-\frac{1}{\lambda_n}[\Vec({\bf W})_{S^c}-\Vec(R({\bf \Delta}))_{S^c}]
\end{align*}
Taking the $l_\infty$ norm, we have
\begin{align*}
    \|{\bf Z}_{S^c}\|_\infty
    &\leq \frac{1}{\lambda_n}\|{\bf \Gamma}^*_{S^cS}({\bf \Gamma}^*_{SS})^{-1}\|_{L_\infty}(\|\Vec({\bf W})_S\|_\infty+\|\Vec(R({\bf \Delta}))_S\|_\infty)\\&+ \|{\bf \Gamma}^*_{S^cS}({\bf \Gamma}^*_{SS})^{-1}\|_{L_\infty}\|\Vec({\bf Z})_S\|_\infty
    +\frac{1}{\lambda_n}(\|\Vec({\bf W})_{S^c}\|_\infty+\|\Vec(R({\bf \Delta}))_{S^c}\|_\infty)\\
    &\leq \frac{2-\alpha}{\lambda_n}(\|{\bf W}\|_\infty+\|R({\bf \Delta})\|_\infty)+1-\alpha
\end{align*}
Since $r\leq \frac{1}{3d(K^*)^4(1+\frac{8}{\alpha})}$,  we know $\|R({\bf \Delta})\|_{\infty}\leq \frac{3}{2}d(K^*)^3 r^2=3d(K^*)^4(1+\frac{8}{\alpha})\frac{8\lambda_n}{\alpha} r\leq \frac{8\lambda_n}{\alpha}$.

Hence, 
\[
    \|{\bf Z}_{S^c}\|_\infty\leq \frac{2-\alpha}{\lambda_n}\frac{2\alpha\lambda_n}{8}+1-\alpha\leq 1-\frac{\alpha}{2}<1
\]
Therefore, we prove that
\[
    \hat{\bf V}_{SGLASSO}=\tilde{\bf V}_{SGLASSO}
\]
As a result, 
\[
    \| \hat{\bf V}_{SGLASSO}-{\bf V}_0 \|_\infty \leq 2K^*(\frac{8}{\alpha}+1) \left(\frac{\sqrt{2}(8 +\eta^2C_{c_0,\eta,s,M })C }{\eta^2}
\sqrt{\frac{\log p }{n}}+\frac{C_{c_0,\eta,s,M }}{\sqrt{p}}\right)
\]
Furthermore, since the support of $\hat {\bf V}_{SGLASSO}$ is the same as that of ${\bf V}_0$, we can know that 
\[
     \|\hat {\bf V}_{SGLASSO}-{\bf V}_0\|_\op\leq \|\hat {\bf V}_{SGLASSO}-{\bf V}_0\|_{L_1}\leq 2K^*d(\frac{8}{\alpha}+1)\left(\frac{\sqrt{2}(8 +\eta^2C_{c_0,\eta,s,M })C }{\eta^2}
\sqrt{\frac{\log p }{n}}+\frac{C_{c_0,\eta,s,M }}{\sqrt{p}}\right)
\] 
And 
\[
    \|\hat {\bf V}_{SGLASSO}-{\bf V}_0\|_{F}^2\leq (p+s)\| \hat{\bf V}_{SGLASSO}-{\bf V}_0 \|_\infty ^2
\]
Therefore, 
\[
        \frac{1}{p}\|\hat {\bf V}_{SGLASSO}-{\bf V}_0\|_{F}^2\leq 4{K^*}^2\left(\frac{s}{p}+1 \right)(\frac{8}{\alpha}+1)^2 \left(\frac{\sqrt{2}(8 +\eta^2C_{c_0,\eta,s,M })C }{\eta^2}
\sqrt{\frac{\log p }{n}}+\frac{C_{c_0,\eta,s,M }}{\sqrt{p}}\right)^2
\]
\end{proof}


\begin{thebibliography}{}

\bibitem[\protect\citeauthoryear{Avella-Medina, Battey, Fan, and
  Li}{Avella-Medina et~al.}{2018}]{avella2018robust}
Avella-Medina, M., H.~S. Battey, J.~Fan, and Q.~Li (2018).
\newblock Robust estimation of high-dimensional covariance and precision
  matrices.
\newblock {\em Biometrika\/}~{\em 105\/}(2), 271--284.

\bibitem[\protect\citeauthoryear{Baba, Shibata, and Sibuya}{Baba
  et~al.}{2004}]{baba2004partial}
Baba, K., R.~Shibata, and M.~Sibuya (2004).
\newblock Partial correlation and conditional correlation as measures of
  conditional independence.
\newblock {\em Australian \& New Zealand Journal of Statistics\/}~{\em
  46\/}(4), 657--664.

\bibitem[\protect\citeauthoryear{Banerjee, Ghaoui, and d'Aspremont}{Banerjee
  et~al.}{2008}]{BanerjeeGhaouidAspremont2008}
Banerjee, O., L.~E. Ghaoui, and A.~d'Aspremont (2008).
\newblock Model selection through sparse maximum likelihood estimation.
\newblock {\em Journal of Machine Learning Research\/}~{\em 9}, 485 -- 516.

\bibitem[\protect\citeauthoryear{Bickel and Levina}{Bickel and
  Levina}{2004}]{190c5a0d-ed5d-320b-82c0-14df50eada30}
Bickel, P. and E.~Levina (2004).
\newblock Some theory for fisher's linear discriminant function, 'naive bayes',
  and some alternatives when there are many more variables than observations.
\newblock {\em Bernoulli\/}~{\em 10\/}(6), 989--1010.

\bibitem[\protect\citeauthoryear{Bickel and Levina}{Bickel and
  Levina}{2008a}]{BickelLevina2008b}
Bickel, P. and E.~Levina (2008a).
\newblock Covariance regularization by thresholding.
\newblock {\em The Annals of Statistics\/}~{\em 36}, 2577 -- 2604.

\bibitem[\protect\citeauthoryear{Bickel and Levina}{Bickel and
  Levina}{2008b}]{BickelLevina2008a}
Bickel, P. and E.~Levina (2008b).
\newblock Regularized estimation of large covariance matrices.
\newblock {\em The Annals of Statistics\/}~{\em 36}, 199 -- 227.

\bibitem[\protect\citeauthoryear{Cai, Liu, and Luo}{Cai
  et~al.}{2011}]{cai2011constrainedl1minimizationapproach}
Cai, T., W.~Liu, and X.~Luo (2011).
\newblock A constrained l1 minimization approach to sparse precision matrix
  estimation.
\newblock {\em Journal of the American Statistical Association\/}~{\em
  494\/}(106), 594--607.

\bibitem[\protect\citeauthoryear{Cai, Ren, and Zhou}{Cai
  et~al.}{2016}]{Cai2016}
Cai, T.~T., Z.~Ren, and H.~H. Zhou (2016).
\newblock Estimating structured high-dimensional covariance and precision
  matrices: Optimal rates and adaptive estimation.
\newblock {\em Electronic Journal of Statistics\/}~{\em 10\/}(1), 1--59.

\bibitem[\protect\citeauthoryear{Chen, Zhang, and Zhao}{Chen
  et~al.}{2011}]{chen2011high}
Chen, X., C.~Zhang, and X.~Zhao (2011).
\newblock High-dimensional linear discriminant analysis: Sparse precision
  matrix estimation and its application.
\newblock {\em Journal of the Royal Statistical Society: Series B (Statistical
  Methodology)\/}~{\em 73\/}(3), 485--506.

\bibitem[\protect\citeauthoryear{d'Aspremont, Banerjee, and
  El~Ghaoui}{d'Aspremont et~al.}{2008}]{dAspremontBanerjeeElGhaoui2008}
d'Aspremont, A., O.~Banerjee, and L.~El~Ghaoui (2008).
\newblock First - order methods for sparse covariance selection.
\newblock {\em SIAM Journal on Matrix Analysis and Its Applications\/}~{\em
  30}, 56 -- 66.

\bibitem[\protect\citeauthoryear{Fan, Feng, and Wu}{Fan
  et~al.}{2009}]{FanFengWu2009}
Fan, J., Y.~Feng, and Y.~Wu (2009).
\newblock Network exploration via the adaptive lasso and scad penalties.
\newblock {\em The Annals of Applied Statistics\/}~{\em 2}, 521 -- 541.

\bibitem[\protect\citeauthoryear{Fan and Li}{Fan and Li}{2001}]{FanLi2001}
Fan, J. and R.~Li (2001).
\newblock Variable selection via nonconcave penalized likelihood and its oracle
  properties.
\newblock {\em Journal of the American Statistical Association\/}~{\em 96},
  1348 -- 1360.

\bibitem[\protect\citeauthoryear{Fan, Liao, and Liu}{Fan
  et~al.}{2016}]{fan2016overview}
Fan, J., Y.~Liao, and H.~Liu (2016).
\newblock An overview of the estimation of large covariance and precision
  matrices.
\newblock {\em The Econometrics Journal\/}~{\em 19\/}(1), C1--C32.

\bibitem[\protect\citeauthoryear{Fang and Anderson}{Fang and
  Anderson}{1990}]{Fang1990}
Fang, K.~T. and T.~W. Anderson (1990).
\newblock {\em Statistical Inference in Elliptically Contoured and Related
  Distributions}.
\newblock New York: Allerton Press Inc.

\bibitem[\protect\citeauthoryear{Fang and Zhang}{Fang and
  Zhang}{1990}]{FangZhang1990}
Fang, K.-T. and Y.-T. Zhang (1990).
\newblock {\em Generalized Multivariate Analysis}.
\newblock Berlin etc.; Beijing: Springer-Verlag; Science Press.

\bibitem[\protect\citeauthoryear{Fang}{Fang}{2018}]{fang2018symmetric}
Fang, K.~W. (2018).
\newblock {\em Symmetric multivariate and related distributions}.
\newblock Chapman and Hall/CRC.

\bibitem[\protect\citeauthoryear{Feng}{Feng}{2024}]{feng2024spatialsignbasedprincipal}
Feng, L. (2024).
\newblock Spatial sign based principal component analysis for high dimensional
  data.
\newblock {\em arXiv\/}, 2409.13267.

\bibitem[\protect\citeauthoryear{Friedman, Hastie, and Tibshirani}{Friedman
  et~al.}{2008}]{FriedmanHastieTibshirani2008}
Friedman, J., T.~Hastie, and R.~Tibshirani (2008).
\newblock Sparse inverse covariance estimation with the graphical lasso.
\newblock {\em Biostatistics\/}~{\em 9}, 432 -- 441.

\bibitem[\protect\citeauthoryear{Han and Liu}{Han and Liu}{2017}]{Han2017}
Han, F. and H.~Liu (2017).
\newblock Statistical analysis of latent generalized correlation matrix
  estimation in transelliptical distribution.
\newblock {\em Bernoulli\/}~{\em 23\/}(1), 23 -- 57.

\bibitem[\protect\citeauthoryear{Han and Liu}{Han and
  Liu}{2018}]{han2016ecahighdimensionalelliptical}
Han, F. and H.~Liu (2018).
\newblock Eca: High-dimensional elliptical component analysis in non-gaussian
  distributions.
\newblock {\em Journal of the American Statistical Association\/}~{\em
  113\/}(521), 252--268.

\bibitem[\protect\citeauthoryear{Huang, Liu, Pourahmadi, and Liu}{Huang
  et~al.}{2006}]{Huangetal2006}
Huang, J., N.~Liu, M.~Pourahmadi, and L.~Liu (2006).
\newblock Covariance matrix selection and estimation via penalized normal
  likelihood.
\newblock {\em Biometrika\/}~{\em 93}, 85 -- 98.

\bibitem[\protect\citeauthoryear{Huang, Li, Li, and Yang}{Huang
  et~al.}{2022}]{huang2022overview}
Huang, Y., C.~Li, R.~Li, and S.~Yang (2022).
\newblock An overview of tests on high-dimensional means.
\newblock {\em Journal of Multivariate Analysis\/}~{\em 188}, 104813.

\bibitem[\protect\citeauthoryear{Lam and Fan}{Lam and Fan}{2009}]{LamFan2009}
Lam, C. and J.~Fan (2009).
\newblock Sparsistency and rates of convergence in large covariance matrix
  estimation.
\newblock {\em The Annals of Statistics\/}~{\em 37}, 4254 -- 4278.

\bibitem[\protect\citeauthoryear{Lauritzen}{Lauritzen}{1996}]{Lauritzen1996}
Lauritzen, S.~L. (1996).
\newblock {\em Graphical Models}.
\newblock Oxford Statistical Science Series. New York: Oxford University Press.

\bibitem[\protect\citeauthoryear{Liu, Han, Yuan, Lafferty, and Wasserman}{Liu
  et~al.}{2012}]{10.1214/12-AOS1037}
Liu, H., F.~Han, M.~Yuan, J.~Lafferty, and L.~Wasserman (2012).
\newblock {High-dimensional semiparametric Gaussian copula graphical models}.
\newblock {\em The Annals of Statistics\/}~{\em 40\/}(4), 2293 -- 2326.

\bibitem[\protect\citeauthoryear{Liu, Roeder, and Wasserman}{Liu
  et~al.}{2010}]{liu2010stability}
Liu, H., K.~Roeder, and L.~Wasserman (2010).
\newblock Stability approach to regularization selection (stars) for high
  dimensional graphical models.
\newblock {\em Advances in neural information processing systems\/}~{\em 23}.

\bibitem[\protect\citeauthoryear{Loh and Tan}{Loh and Tan}{2018}]{Loh2018}
Loh, P.~L. and X.~L. Tan (2018).
\newblock High - dimensional robust precision matrix estimation: cellwise
  corruption under $\varepsilon$-contamination.
\newblock {\em Electronic Journal of Statistics\/}~{\em 12\/}(1), 1429 -- 1467.

\bibitem[\protect\citeauthoryear{McLachlan}{McLachlan}{2004}]{mclachlan2004discriminant}
McLachlan, G. (2004).
\newblock {\em Discriminant Analysis and Statistical Pattern Recognition}.
\newblock Wiley-Interscience.

\bibitem[\protect\citeauthoryear{Oja}{Oja}{2010}]{oja2010multivariate}
Oja, H. (2010).
\newblock {\em Multivariate nonparametric methods with R: an approach based on
  spatial signs and ranks}.
\newblock Springer Science \& Business Media.

\bibitem[\protect\citeauthoryear{Ollerer and Croux}{Ollerer and
  Croux}{2015}]{Ollerer2015}
Ollerer, V. and C.~Croux (2015).
\newblock Robust high - dimensional precision matrix estimation.
\newblock In {\em Modern Nonparametric, Robust and Multivariate Methods}, pp.\
  325 -- 350. Cham: Springer.

\bibitem[\protect\citeauthoryear{Owen and Rabinovitch}{Owen and
  Rabinovitch}{1983}]{https://doi.org/10.1111/j.1540-6261.1983.tb02499.x}
Owen, J. and R.~Rabinovitch (1983).
\newblock On the class of elliptical distributions and their applications to
  the theory of portfolio choice.
\newblock {\em The Journal of Finance\/}~{\em 38\/}(3), 745--752.

\bibitem[\protect\citeauthoryear{Qin}{Qin}{2018}]{qin2018review}
Qin, Y. (2018).
\newblock A review of quadratic discriminant analysis for high-dimensional
  data.
\newblock {\em Wiley Interdisciplinary Reviews: Computational
  Statistics\/}~{\em 10\/}(4), e1434.

\bibitem[\protect\citeauthoryear{Ravikumar, Wainwright, Raskutti, and
  Yu}{Ravikumar et~al.}{2011}]{RavikumarWainwrightRaskuttiYu2008}
Ravikumar, P., M.~J. Wainwright, G.~Raskutti, and B.~Yu (2011).
\newblock {High-dimensional covariance estimation by minimizing $l_1$-penalized
  log-determinant divergence}.
\newblock {\em Electronic Journal of Statistics\/}~{\em 5}, 935 -- 980.

\bibitem[\protect\citeauthoryear{Rothman, Bickel, Levina, and Zhu}{Rothman
  et~al.}{2008}]{Rothmanetal2008}
Rothman, A., P.~Bickel, E.~Levina, and J.~Zhu (2008).
\newblock Sparse permutation invariant covariance estimation.
\newblock {\em electronic Journal of Statistics\/}~{\em 2}, 494 -- 515.

\bibitem[\protect\citeauthoryear{Shao, Wang, Deng, and Wang}{Shao
  et~al.}{2011}]{Shao2011}
Shao, J., Y.~Wang, X.~Deng, and S.~Wang (2011).
\newblock Sparse linear discriminant analysis with high dimensional data.
\newblock {\em The Annals of Statistics\/}~{\em 39}, 1241--1265.

\bibitem[\protect\citeauthoryear{Tarr, Müller, and Weber}{Tarr
  et~al.}{2016}]{Tarr2016}
Tarr, G., S.~Müller, and N.~C. Weber (2016).
\newblock Robust estimation of precision matrices under cellwise contamination.
\newblock {\em Computational Statistics and Data Analysis\/}~{\em 93}, 404 --
  420.

\bibitem[\protect\citeauthoryear{Vershynin}{Vershynin}{2018}]{Vershynin_2018}
Vershynin, R. (2018).
\newblock {\em High-Dimensional Probability: An Introduction with Applications
  in Data Science}.
\newblock Cambridge Series in Statistical and Probabilistic Mathematics.
  Cambridge University Press.

\bibitem[\protect\citeauthoryear{Vogel and Fried}{Vogel and
  Fried}{2011}]{vogel2011elliptical}
Vogel, D. and R.~Fried (2011).
\newblock Elliptical graphical modelling.
\newblock {\em Biometrika\/}~{\em 98\/}(4), 935--951.

\bibitem[\protect\citeauthoryear{Vogel and Tyler}{Vogel and
  Tyler}{2014}]{vogel2014robust}
Vogel, D. and D.~E. Tyler (2014).
\newblock Robust estimators for nondecomposable elliptical graphical models.
\newblock {\em Biometrika\/}~{\em 101\/}(4), 865--882.

\bibitem[\protect\citeauthoryear{Von~Roemeling, Radisky, Marlow, Cooper, Grebe,
  Anastasiadis, Tun, and Copland}{Von~Roemeling et~al.}{2014}]{von2014neuronal}
Von~Roemeling, C.~A., D.~C. Radisky, L.~A. Marlow, S.~J. Cooper, S.~K. Grebe,
  P.~Z. Anastasiadis, H.~W. Tun, and J.~A. Copland (2014).
\newblock Neuronal pentraxin 2 supports clear cell renal cell carcinoma by
  activating the ampa-selective glutamate receptor-4.
\newblock {\em Cancer research\/}~{\em 74\/}(17), 4796--4810.

\bibitem[\protect\citeauthoryear{Wang, Xie, and Kang}{Wang
  et~al.}{2023}]{wang2023novel}
Wang, S., C.~Xie, and X.~Kang (2023).
\newblock A novel robust estimation for high-dimensional precision matrices.
\newblock {\em Statistics in Medicine\/}~{\em 42\/}(5), 656--675.

\bibitem[\protect\citeauthoryear{Wegkamp and Zhao}{Wegkamp and
  Zhao}{2016}]{Wegkamp2016}
Wegkamp, M. and Y.~Zhao (2016).
\newblock Adaptive estimation of the copula correlation matrix for
  semiparametric elliptical copulas.
\newblock {\em Bernoulli\/}~{\em 22\/}(2), 1184 -- 1226.

\bibitem[\protect\citeauthoryear{Whittaker}{Whittaker}{2009}]{whittaker2009graphical}
Whittaker, J. (2009).
\newblock {\em Graphical models in applied multivariate statistics}.
\newblock Wiley Publishing.

\bibitem[\protect\citeauthoryear{Wu and Pourahmadi}{Wu and
  Pourahmadi}{2003}]{WuPourahmadi2003}
Wu, W.~B. and M.~Pourahmadi (2003).
\newblock Nonparametric estimation of large covariance matrices of longitudinal
  data.
\newblock {\em Biometrika\/}~{\em 90}, 831 -- 844.

\bibitem[\protect\citeauthoryear{Yuan}{Yuan}{2009}]{Yuan2009}
Yuan, M. (2009).
\newblock Sparse inverse covariance matrix estimation via linear programming.
\newblock {\em Journal of Machine Learning Research\/}~{\em 11}, 2261 -- 2286.

\bibitem[\protect\citeauthoryear{Yuan and Lin}{Yuan and
  Lin}{2007}]{YuanLin2007}
Yuan, M. and Y.~Lin (2007).
\newblock Model selection and estimation in the gaussian graphical model.
\newblock {\em Biometrika\/}~{\em 94}, 19 -- 35.

\bibitem[\protect\citeauthoryear{Zhang}{Zhang}{2010}]{Zhang2010}
Zhang, C.-H. (2010).
\newblock Nearly unbiased variable selection under minimax concave penalty.
\newblock {\em Annals of Statistics\/}~{\em 38}, 894--942.

\bibitem[\protect\citeauthoryear{Zhang, Bai, Yan, and Li}{Zhang
  et~al.}{2020}]{zhang2020high}
Zhang, Z., X.~Bai, J.~Yan, and L.~Li (2020).
\newblock High-dimensional linear discriminant analysis: A review.
\newblock {\em Statistical Methods in Medical Research\/}~{\em 29\/}(2),
  356--380.

\end{thebibliography}

\end{document}